\documentclass[twocolumn]{autart}    

\usepackage{graphicx}          
                               
\usepackage{amsmath}
\usepackage{mathrsfs}
\usepackage{amssymb}
\usepackage{color}
\usepackage{cite}
\usepackage{comment}
\usepackage{float}
\usepackage{graphicx}
\usepackage{float}
\usepackage{booktabs}
\usepackage{longtable,tabularx}
\usepackage{nicefrac}
\usepackage{mathtools}
\usepackage{bm}
\usepackage{multirow}
\usepackage{apptools}
\usepackage{algorithm}
\usepackage[noend]{algpseudocode}
\usepackage{ulem}
\usepackage{indentfirst}
\usepackage{siunitx}

\AtAppendix{\counterwithin{lma}{section}}

\newcommand{\isdef}{\triangleq}

\newcommand{\eigmax}{\boldsymbol{\lambda_{\rm max}}}
\newcommand{\eigmin}{\boldsymbol{\lambda_{\rm min}}}
\newcommand{\svdmax}{\boldsymbol{\sigma_{\rm max}}}
\newcommand{\svdmin}{\boldsymbol{\sigma_{\rm min}}}

\newcommand{\rmT}{{\rm T}}

\newcommand{\BBR}{{\mathbb R}}

\newtheorem{lma}{Lemma}
\newtheorem{theo}{Theorem}

\newenvironment{proof}{\textbf{Proof.}}{\hfill$\blacksquare$}

\DeclareMathOperator*{\argmin}{arg\,min}
\DeclareMathOperator{\rank}{rank}

\DeclareMathOperator{\diag}{diag}

\begin{document}

\begin{frontmatter}

\title{SIFt-RLS: Subspace of Information Forgetting Recursive Least Squares} 


\author[UM]{Brian Lai}\ead{brianlai@umich.edu} and \author[UM]{Dennis S. Bernstein}\ead{dsbaero@umich.edu}           

\address[UM]{Department of Aerospace Engineering, The University of Michigan, Ann Arbor, Michigan 48109, United States}  

\begin{keyword}                           
    System Identification; Recursive Identification; Parameter Estimation; Adaptive Systems; Directional Forgetting                                         
\end{keyword}                             

\begin{abstract}                          
This paper presents subspace of information forgetting recursive least squares (SIFt-RLS), a directional forgetting algorithm which, at each step, forgets only in row space of the regressor matrix, or the \textit{information subspace}.
As a result, SIFt-RLS tracks parameters that are in excited directions while not changing parameter estimation in unexcited directions.
It is shown that SIFt-RLS guarantees an upper and lower bound of the covariance matrix, without assumptions of persistent excitation, and explicit bounds are given.
Furthermore, sufficient conditions are given for the uniform Lyapunov stability and global uniform exponential stability of parameter estimation error in SIFt-RLS when estimating fixed parameters without noise.
SIFt-RLS is compared to other RLS algorithms from the literature in a numerical example without persistently exciting data.
\end{abstract}

\end{frontmatter}

\section{Introduction}

Recursive least squares (RLS) is a fundamental algorithm in systems and control theory for the online identification of fixed parameters \cite{ljung1983theory,aastrom2013adaptive}.
An important property of RLS is that the eigenvalues of the covariance matrix are monotonically decreasing and may become arbitrarily small \cite[subsection 2.3.2]{sastry1990adaptive}, \cite{goel2020recursive}, resulting in eventual slow adaptation and inability to track time-varying parameters \cite{ortega2020modified,salgado1988modified,lai2021regularization}.
A classical extension of RLS used to track time-varying parameters is exponential forgetting \cite{johnstone1982exponential,islam2019recursive}, which uses a forgetting factor to put exponentially lower weights on older terms in the least squares cost function. 
While this allows for continued adaptation of parameter estimation, a critical issue that arises is, without persistent excitation, at least one of the eigenvalues the covariance matrix may become arbitrarily large \cite{johnstone1982exponential,goel2020recursive}, a phenomenon known as covariance blow up \cite{malik1991some} or covariance windup \cite[p. 473]{aastrom2013adaptive}.
This results in significant sensitivity to measurement noise and poor parameter estimation \cite{goel2020recursive}.

As such, a well-established heuristic for RLS extensions is that the eigenvalues of the covariance matrix should have an upper bound and positive lower bound \cite{salgado1988modified}.
Later works have shown that conditions related to the boundedness of the covariance matrix are sufficient to guarantee stability of the parameter estimation error \cite{parkum1992recursive, lai2023generalized, shaferman2021continuous}. 
Common approaches to bound the covariance matrix and further improve tracking of time-varying parameters include variable-rate forgetting, in which a time-varying forgetting factor is used \cite{bruce2020convergence,fortescue1981implementation,Dasgupta1987Asymptotically,paleologu2008robust,Hung2005Gradient,mohseni2022recursive}, and resetting strategies, in which the covariance matrix is periodically reset to a desired value \cite{goodwin1983deterministic} or slowly converges back to a desired value when there is little excitation \cite{salgado1988modified,shin2020new,Lai2023Exponential}.
These methods address the fact that there may be periods of rich excitation, in which forgetting old data can more effectively adapt parameter estimates, and periods of poor excitation, in which it is preferable to rely on old data for parameter estimation.

While variable-rate forgetting and resetting strategies address periods of rich and poor excitation, these methods do not address when excitation is nonuniform. 
Nonuniform excitation may result in sufficient excitation to quickly adapt parameters in particular \textit{directions} (that is, particular linear combinations of the parameters), while parameters in other directions can only be slowly adapted or not adapted at all \cite{vahidi2005recursive,kulhavy1984tracking,salgado1988modified}.
If it is known a priori how quickly parameters will vary in different directions and which directions will be excited, techniques such as multiple forgetting may be suitable, where different forgetting factors are chosen for different directions \cite{vahidi2005recursive,fraccaroli2015new}. 
However, a more challenging problem is forgetting when such information is not known a priori. 
To tackle this problem, various directional forgetting techniques have been developed which analyze the regressor and/or the covariance matrix to gauge which directions are being excited and which are not \cite{cao2000directional,kulhavy1984tracking,bittanti1990convergence,zhu2021recursive,goel2020recursive}.

This work presents a new directional forgetting algorithm called subspace of information forgetting (SIFt) RLS.
We begin by presenting a subspace decomposition of a positive definite matrix into the sum of two positive semidefinite matrices, one \textit{parallel} to the subspace and one \textit{orthogonal} to the subspace.
This decomposition reduces to the decomposition used in \cite{cao2000directional} when the subspace is of dimension $1$.
We give a thorough analysis of this decomposition, including the main properties of the decomposition, the uniqueness of this decomposition, and a duality between the parallel and orthogonal components.

Next, we develop the SIFt-RLS algorithm, where, at each step, we call the row space of the regressor the \textit{information subspace}.
At each step, the inverse covariance matrix (also called the information matrix) is decomposed into the components parallel and and orthogonal to the information subspace, and forgetting is applied only to the parallel component. 
This approach is developed for vector measurements and specializes to the method in \cite{cao2000directional} in the case of scalar measurements. 
Forgetting techniques for parameter estimation with vector measurements are useful, for example, in adaptive control of multiple-input, multiple-output (MIMO) systems \cite{nguyen2021predictive,islam2019recursive}.

Finally, we give explicit upper and lower bounds for the eigenvalues of the covariance matrix in SIFt-RLS that are guaranteed without persistent excitation.
This goes beyond the analysis of \cite{cao2000directional} which only shows the existence of bounds.
Moreover, we show a striking similarity between the eigenvalue bounds of exponential forgetting that are guaranteed with persistent excitation and the eigenvalue bounds of SIFt.
This shows that forgetting in the information subspace yields analogous bounds to uniform forgetting when all directions are excited.
Furthermore, a counterexample shows that \cite{zhu2021recursive}, a different vector measurement extension of \cite{cao2000directional}, does not guarantee an upper bound on the eigenvalues of the covariance matrix as SIFt does.
Finally, we provide sufficient conditions for the uniform Lyapunov stability and global uniform exponential stability of the estimation error dynamics of SIFt-RLS.

\subsection{Notation and Terminology}
$I_n$ denotes the $n \times n$ identity matrix, and $0_{m \times n}$ denotes the $m \times n$ zero matrix. 
For $B \in \BBR^{m \times n}$, $\svdmax(B)$ denotes the largest singular value of $B$, and $\svdmin(B)$ denotes the smallest singular value of $B$.
For symmetric $A\in \BBR^{n \times n}$, let the $n$ real eigenvalues of $A$ be denoted by
$\eigmin(A) \triangleq \boldsymbol{\lambda_n}(A) \le \cdots \le \eigmax(A) \triangleq \boldsymbol{\lambda_1}(A)$. 
Additionally, the spectral radius of $A$, $\rho(A)$, is defined as 
\begin{align}
    \rho(A) \triangleq \max \{ \vert \boldsymbol{\lambda_1}(A) \vert, \cdots, \vert \boldsymbol{\lambda_n}(A) \vert\}.
\end{align}
Furthermore, if A is positive definite, the singular values of $A$ are equivalent to the eigenvalues of $A$ and the condition number of A, $\kappa (A)$, can be expressed as
\begin{align}
    \kappa(A) \isdef \frac{\svdmax(A)}{\svdmin(A)} = \frac{\eigmax(A)}{\eigmin(A)}.
    \label{eq::Cond_numb_Defn}
\end{align}
If $A$ is positive semidefinite but not positive definite, then $\kappa(A)\isdef\infty.$  

For $x \in \BBR^n$, $\Vert x \Vert$ denotes the Euclidean norm of $x$. For $B \in \BBR^{m \times n}$, $B^+$ denotes the Moore-Penrose inverse of $B$.
$\mathcal{R}(B)$ and $\mathcal{N}(B)$ denote the column space and nullspace of $B$, respectively.
For positive-semidefinite $R \in \BBR^{n \times n}$ and $x \in \BBR^n$, $\Vert x \Vert_R \triangleq \sqrt{x^\rmT R x}$. 
Let $( a_k )_{k=0}^\infty \subset \BBR^{m \times n}$ denote that, for all $k \ge 0$, $a_k \in \BBR^{m \times n}$. For $P,Q \in \BBR^{n \times n}$, let $P \prec Q$ (respectively $P \preceq Q$) denote that $Q-P$ is positive definite (respectively positive semidefinite). 

\begin{defn}
A sequence $(\phi_k)_{k=0}^\infty \subset \BBR^{p \times n}$ is \textbf{persistently exciting} if there exist $N \ge 1$ and $\alpha>0$ 
such that, for all $k \ge 0$,
\begin{align}
    \alpha I_n \preceq \sum_{i=k}^{k+N}\phi_i^\rmT \phi_i.\label{eqn: persistent excitation defn}
\end{align}
Furthermore, $\alpha$ and $N$ are, respectively, the \textbf{lower bound}
 and \textbf{persistency window} of $(\phi_k)_{k=0}^\infty$.
\end{defn}

\begin{defn}
\label{defn: sequence upper bounded}
A sequence $(\phi_k)_{k=0}^\infty \subset \BBR^{p \times n}$ is \textbf{bounded} if there exist $\beta>0$ 
such that, for all $k \ge 0$,
\begin{align}
    \phi_k^\rmT \phi_k \preceq \beta I_n.\label{eqn: bounded excitation defn}
\end{align}
Furthermore, $\beta$ is the \textbf{upper bound} of $(\phi_k)_{k=0}^\infty$.
\end{defn}

\section{RLS with Exponential Forgetting}

To begin, we introduce RLS with exponential forgetting in Proposition \ref{prop: RLS}, which gives a recursive way to compute the global minimizer of a least squares cost with exponentially smaller weighting on older terms.
This background will be useful for later comparison with SIFt-RLS.

\begin{prop}
    \label{prop: RLS}
    For all $k \ge 0$, let $\phi_k \in \BBR^{p \times n}$ and $y_k \in \BBR^p$. Furthermore, let $\theta_0 \in \BBR^n$, let $R_0 \in \BBR^{n \times n}$ be positive definite, and let $\lambda \in (0,1)$.
    For all $k \ge 0$, define the function $J_k \colon \BBR^n \rightarrow \BBR$ by
    \begin{align}
        J_k(\hat{\theta}) \triangleq \lambda^{k+1} \Vert \hat{\theta} - \theta_0 \Vert_{R_0}^2 + \sum_{i=0}^k  \lambda^{k-i} \Vert y_i - \phi_i \hat{\theta} \Vert_2.
    \end{align}
    Then, $J_k$ has a unique global minimizer, denoted 
    \begin{align}
        \theta_{k+1} \triangleq \argmin_{\hat{\theta} \in \BBR^n} J_k(\hat{\theta}),
    \end{align}
    which, for all $k \ge 0$, is given recursively by
    \begin{align}
        R_{k+1} &= \lambda R_k + \phi_k^\rmT \phi_k, \label{eqn: R_k update} \\
        \theta_{k+1} &= \theta_k + P_{k+1} \phi_k^\rmT(y_k-\phi_k \theta_k), \label{eqn: theta_k update}
    \end{align}
    where, for all $k \ge 0$, $R_k \in \BBR^{n \times n}$ is positive definite and $P_k \triangleq R_k^{-1} \in \BBR^{n \times n}$.
    Moreover, for all $k \ge 0$, $P_{k+1}$ is given recursively by
    \begin{align}
        P_{k+1} = \frac{1}{\lambda} P_k - \frac{1}{\lambda} P_k \phi_k^\rmT (\lambda I_p + \phi_k P_k \phi_k^\rmT)^{-1}\phi_k P_k. \label{eqn: P_k update}
    \end{align}
\end{prop}
\begin{proof}
    See Theorem 1 of \cite{islam2019recursive}.
\end{proof}

For all $k \ge 0$, We call 
$y_k \in \BBR^p$ the \textit{measurement}, 
$\phi_k \in \BBR^{p \times n}$ the \textit{regressor}, 
and $\theta_k \in \BBR^n$ the \textit{parameter estimate}. 
Moreover, we call 
$R_k \in \BBR^{n \times n}$ is the \textit{information matrix}, 
and $P_k \in \BBR^{n \times n}$ the \textit{covariance matrix}.
Finally, we call $\lambda \in (0,1)$ the \textit{forgetting factor}.

Note that RLS with exponential forgetting may experience covariance windup, where some of the eigenvalues of the covariance matrix $P_k$ may become unbounded as $k \rightarrow \infty$ (e.g. see example 3 of \cite{goel2020recursive}).
However, if the sequence of regressors $(\phi_k)_{k=0}^\infty$ is persistently exciting and bounded, then $P_k$ has guaranteed upper and lower bounds \cite{goel2020recursive}.
Proposition \ref{prop: EF RLS Covariance Bound with N=1} gives an expression for these bounds in the special case where the persistency window is $N=1$. For bounds in the case $N > 1$, see Proposition 4 of \cite{goel2020recursive}.
\begin{prop}
\label{prop: EF RLS Covariance Bound with N=1}
Let $\lambda \in (0,1)$, let $\theta_0 \in \BBR^n$, and let $R_0 \in \BBR^{n \times n}$ be positive definite. 
For all $k \ge 0$, let $y_k \in \BBR^p$, and $\phi_k \in \BBR^{p \times n}$. Finally, for all $k \ge 0$, let $R_{k} \in \BBR^{n \times n}$ and $\theta_k \in \BBR^n$ be recursively updated by \eqref{eqn: R_k update} and \eqref{eqn: theta_k update} and $P_k \triangleq R_k^{-1}$.
Then the following statements hold:
\begin{enumerate}
    \item If $\left( \phi_k \right)_{k=0}^\infty$ is persistently exciting with lower bound $\alpha > 0$ and persistency window $N = 1$, then, for all $k \ge 0$,
    \begin{align}
        \eigmin(R_k) &\ge \min \{ \frac{\alpha}{1-\lambda}, \eigmin(R_0) \} ,
        \label{eqn: RLS Rk lower bound}
        \\
        \eigmax(P_k) &\le \max \{ \frac{1-\lambda}{\alpha} , \eigmax(P_0) \} .
        \label{eqn: RLS Pk upper bound}
    \end{align}
    \item If $\left( \phi_k \right)_{k=0}^\infty$ is bounded with upper bound $\beta \in (0,\infty)$, then, for all $k \ge 0$, 
    \begin{align}
        \eigmax(R_k) \le \max \{ \frac{\beta}{1-\lambda}, \eigmax(R_0) \},
        \label{eqn: RLS Rk upper bound}
        \\
        \eigmin(P_k) \ge \min \{ \frac{1-\lambda}{\beta}, \eigmin(P_0) \}.
        \label{eqn: RLS Pk lower bound}
    \end{align}
\end{enumerate}
\end{prop}
\begin{proof}
    We first prove \textit{(1)} by induction. For brevity, let $a \triangleq \min \{ \frac{\alpha}{1-\lambda}, \eigmin(R_0) \}$. For the base case, note that $\eigmin(R_0) \ge a$ immediately holds. 
    Next, assume for inductive hypothesis that $\eigmin(R_k) \ge a$. 
    It follows from \eqref{eqn: R_k update}, Lemma \ref{lem: lambda_k(A) + lambda_1(B) < lambda_k(A+B) < lambda_k(A) + lambda_n(B)}, and persistent excitation of $\left( \phi_k \right)_{k=0}^\infty$ that $\eigmin(R_{k+1}) \ge \lambda a + \alpha$. 
    Substituting $a = \min \{ \frac{\alpha}{1-\lambda}, \eigmin(R_0) \}$ yields $\eigmin(R_{k+1}) \ge \min \{\frac{\alpha}{1-\lambda} , \lambda \eigmin(R_0) + \alpha\} I_n$.
    
    In the case where $\eigmin(R_0) \ge \frac{\alpha}{1-\lambda}$, it follows that $\lambda \eigmin(R_0) + \alpha \ge \lambda (\frac{\alpha}{1-\lambda}) + \alpha = \frac{\alpha}{1-\lambda} \ge a$.
    In the case where $\eigmin(R_0) < \frac{\alpha}{1-\lambda}$, it follows that $\lambda \eigmin(R_0) + \alpha \ge \eigmin(R_0) \ge a$. 
    Hence, $R_{k+1} \succeq a I_n$, and \eqref{eqn: RLS Rk lower bound} holds by mathematical induction.
    Lastly \eqref{eqn: RLS Pk upper bound} follows directly from \eqref{eqn: RLS Rk lower bound} since, for all $k \ge 0$, $P_k = R_k^{-1}$.
    Statement \textit{(2)} can be proven using similar reasoning.
\end{proof}

\section{Subspace Decomposition of a Positive-Definite Matrix}
\label{sec: subspace decomp}
This section presents a novel decomposition of a positive-definite matrix $A \in \BBR^{n \times n}$ into the sum of two positive-semidefinite matrices based on the choice of a subspace $S \subset \BBR^n$. 
This decomposition will be the main tool used in SIFt-RLS. 
Note that this decomposition reduces to the decomposition proposed in \cite{cao2000directional} when subspace $S$ is dimension $1$.
To begin, Definition \ref{defin: subspace decomposition} defines the proposed matrix decomposition.

\begin{defn}
\label{defin: subspace decomposition}
Let $S \subset \BBR^n$ be a subspace of dimension $p \le n$ and let $A \in \BBR^{n \times n}$ be positive definite. Let $v_1, \ldots ,v_p \in \BBR^n$ be a basis for $S$ and $v \triangleq \begin{bmatrix} v_1 & \cdots & v_p \end{bmatrix} \in \BBR^{n \times p}$. Then, define the operators $\parallel$ and $\perp$ as
\begin{align}
    A^{\parallel S} &\triangleq  Av(v^T A v)^{-1} v^T A \in \BBR^{n \times n}, \label{eqn: Aparallel definition} \\
    A^{\perp S} &\triangleq  A- A^{\parallel S} \in \BBR^{n \times n}, \label{eqn: Aperp definition}
\end{align}
where, by Lemma \ref{lem: phi R phi^T PSD}, $v^\rmT A v \in \BBR^{p \times p}$ is nonsingular. 
\end{defn}

We call $A^{\parallel S}$ ``$A$ parallel to $S$" and $A^{\perp S}$ ``$A$ orthogonal to $S$". We call the equality $A = A^{\parallel S} + A^{\perp S}$ the ``$S$ subspace decomposition of $A$".
Next, Proposition \ref{prop: Apar change of basis} shows that the matrices $A^{\parallel S}$ and $A^{\perp S}$ do not depend on the choice of basis for $S$.
\begin{prop}
\label{prop: Apar change of basis}
Let $S \subset \BBR^n$ be a subspace of dimension $p \le n$ and let $A \in \BBR^{n \times n}$ be positive definite. Let $v_1, \cdots ,v_p \in \BBR^n$ and $w_1, \cdots ,w_p \in \BBR^n$ be a bases for $S$ and $v \triangleq \begin{bmatrix} v_1 & \cdots & v_p \end{bmatrix} \in \BBR^{n \times p}$, $w \triangleq \begin{bmatrix} w_1 & \cdots & w_p \end{bmatrix} \in \BBR^{n \times p}$. Then,
\begin{align}
    Av(v^\rmT A v)^{-1} v^\rmT A = Aw(w^\rmT A w)^{-1} w^\rmT A. \label{eqn: Apar change of basis}
\end{align}
\end{prop}
\begin{proof}
    Since $v$ and $w$ are both bases for $S$, there exists invertible $P \in \BBR^{p \times p}$ such that $v = wP$. Then, \eqref{eqn: Apar change of basis} follows from substituting $v = wP$ in $Av(v^\rmT A v)^{-1} v^\rmT A$.
\end{proof}

\subsection{Properties}

Theorem \ref{theo: Apar Aperp properties} now gives the main properties of the matrices $A^{\parallel S}$ and $A^{\perp S}$ in the $S$ subspace decomposition of $A$.
\begin{theo}[Properties of $A^{\parallel S}$ and $A^{\perp S}$]
\label{theo: Apar Aperp properties} 
Let $S \subset \BBR^n$ be a subspace of dimension $p \le n$ and let $A \in \BBR^{n \times n}$ be positive definite. Then, $A$ can be decomposed into 
\begin{align}
    A = A^{\parallel S} + A^{\perp S},
\end{align}
where $A^{\parallel S} \in \BBR^{n \times n}$, given by \eqref{eqn: Aparallel definition}, satisfies the properties:
\begin{enumerate}
    \item[(1a)] $A^{\parallel S}$ is positive semidefinite,
    \item[(2a)] $\rank(A^{\parallel S}) = p$,
    \item[(3a)] For all $x \in S$, $A^{\parallel S}x = Ax$.
\end{enumerate}
and $A^{\perp S} \in \BBR^{n \times n}$, given by \eqref{eqn: Aperp definition}, satisfies the properties:
\begin{enumerate}
    \item[(1b)] $A^{\perp S}$ is positive semidefinite,
    \item[(2b)] $\rank(A^{\perp S}) = n-p$,
    \item[(3b)] For all $x \in S$, $A^{\perp S}x = 0$.
\end{enumerate}
\end{theo}

\begin{proof}
    See Appendix Section \ref{sec: subspace decomp proofs}.
\end{proof}

Note that properties \textit{2a)} and \textit{3a)} imply that $\mathcal{R}(A^{\parallel S}) = AS$ and properties \textit{2b)} and \textit{3b)} imply that $\mathcal{N}(A^{\perp S}) = S$.

Next, Theorem \ref{theo: A^parallel uniqueness} shows that $A^{\parallel S}$ is the unique matrix satisfying properties \textit{(1a*)}, \textit{(2a)}, and \textit{(3a)} of Theorem \ref{theo: A^parallel uniqueness}.
Note that since property \textit{(1a)} of Theorem \ref{theo: Apar Aperp properties} implies property \textit{(1a*)}, $A^{\parallel S}$ is also the unique matrix satisfying properties \textit{(1a)}, \textit{(2a)}, and \textit{(3a)} of Theorem \ref{theo: Apar Aperp properties}.

\begin{theo}[Uniqueness of $A^{\parallel S}$]
\label{theo: A^parallel uniqueness}
Let $S \subset \BBR^n$ be a subspace of dimension $p \le n$ and $A \in \BBR^{n \times n}$ be positive definite. If $\tilde{A} \in \BBR^{n \times n}$ satisfies the properties
\begin{enumerate}
    \item[(1a*)] $\tilde{A}$ is symmetric,
    \item[(2a)] $\rank(\tilde{A}) = p$,
    \item[(3a)] For all $x \in S$, $\tilde{A}x = A x$,
\end{enumerate}
then $\tilde{A} = A^{\parallel S}$.
\end{theo}

\begin{proof}
    See Appendix Section \ref{sec: subspace decomp proofs}.
\end{proof}


Next, Definition \ref{defn: A inner product} defines the $\langle \cdot , \cdot \rangle_A$ inner product for positive-definite matrix $A$ and Definition \ref{defn: orthogonal complement in A inner product} defines $S^{\perp_A}$, the orthogonal complement of subspace $S$ under this inner product.
Finally, Proposition \ref{prop: duality} shows a duality between the $\parallel$ and $\perp$ operators and the subspaces $S$ and $S^{\perp_A}$.


\begin{defn}
\label{defn: A inner product}
Let $A \in \BBR^{n \times n}$ be positive definite. Define the inner product $\langle \cdot , \cdot \rangle_A \colon \BBR^{n} \times \BBR^n \rightarrow \BBR$ as
\begin{align}
    \langle x,y \rangle_A \triangleq x^\rmT A y.
\end{align}
\end{defn}

\begin{defn}
\label{defn: orthogonal complement in A inner product}
Let $S \subset \BBR^n$ be a subspace and let $A \in \BBR^{n \times n}$ be positive definite. Define the orthogonal complement of $S$ under the inner product $\langle \cdot , \cdot \rangle_A$ as
\begin{align}
    S^{\perp_A} \triangleq \left\{x \in \BBR^m : \textnormal{for all $y \in S$, } \langle x,y \rangle_A = 0 \right\}.
\end{align}
\end{defn}


%

\begin{prop}[Duality]
\label{prop: duality}
Let $S \subset \BBR^n$ be a subspace of dimension $p \le n$ and let $A \in \BBR^{n \times n}$ be positive definite. 
%
Then,
\begin{align}
    A^{\parallel S} &= 
    A^{\perp (S^{\perp_A})} , \label{eqn: AparS duality} \\
    A^{\perp S} &= 
    A^{\parallel (S^{\perp_A})}. \label{eqn: AperpS duality}
\end{align}
\end{prop}
\begin{proof}
    See Appendix Section \ref{sec: subspace decomp proofs}.
\end{proof}

\section{SIFt-RLS}

To motivate this algorithm, suppose there exists true parameters $\theta \in \BBR^n$ such that, for all $k \ge 0$, $y_k = \phi_k \theta$. 
Then, at any $k \ge 0$, knowledge of the regressor $\phi_k \in \BBR^{p \times n}$ and measurement $y_k \in \BBR^p$ only provide information about the true parameters in the row space of $\phi_k$. 
For example, a regressor $\phi_k = [1 \ 1 \ 0]$ only informs about the sum of the first two parameters of $\theta \in \BBR^3$.
We call the row space of $\phi_k$, i.e. $\mathcal{R}({\phi}_k^\rmT)$, the \textit{information subspace}.

Note that in \eqref{eqn: R_k update}, $R_k$ is multiplied by $\lambda \in (0,1)$, which can be interpreted as forgetting uniformly over all directions.
Covariance windup in exponential forgetting RLS occurs when the information subspace does not contain particular directions over a large number of time steps, while forgetting is uniform.
As a result, persistent excitation conditions are often used to guarantee excitation in all directions.

We propose subspace of information forgetting recursive least squares, or SIFt-RLS, a directional forgetting algorithm which, at each step $k \ge 0$, only forgets in the information subspace of that step.
To summarize the algorithm, at each step $k \ge 0$, SIFt-RLS involves three parts: 
\begin{enumerate}
    \item[1.] \textit{Information Filtering}: First, we take a compact singular value decomposition (SVD) of the regressor $\phi_k \in \BBR^{ p \times n}$ to truncate any singular values smaller than a desired threshold. 
    This gives the filtered regressor $\bar{\phi}_k$ and filtered measurement $\bar{y}_k$.
    This step is critical to ensure the algorithm is numerically stable. 
    \item[2.] \textit{Subspace of Information Forgetting (SIFTing)}: Next, we decompose the information matrix $R_k$ using the subspace decomposition presented in section \ref{sec: subspace decomp} with the information subspace $\mathcal{R}(\bar{\phi}_k^\rmT)$.
    Forgetting is applied only to the component of $R_k$ parallel to the information subspace.
    \item[3.] \textit{Update}: Lastly, we compute the updated information matrix $R_{k+1}$ and parameter estimate $\theta_{k+1}$ using the filtered regressor and filtered measurement.
\end{enumerate} 

Furthermore, we also provide a way to directly update the covariance matrix $P_k \triangleq R_k^{-1}$ in steps 2 and 3 using the matrix inversion Lemma, which is more computationally efficient when $p \ll n$.
The next three subsections provide details on the three parts of SIFt-RLS. 
The algorithm is initialized with forgetting factor $\lambda \in (0,1)$, tuning parameter $\varepsilon > 0$, initial estimate of parameters $\theta_0 \in \BBR^n$, initial positive-definite information matrix $R_0 \in \BBR^{n \times n}$, and initial covariance matrix $P_0^{-1}$.

\subsection{Information Filtering}

To begin, let $\varepsilon > 0$ be a tuning parameter which, for all $k \ge 0$, will be used to truncate small singular values of $\phi_k$.
Next, let $\phi_k = U_k \Sigma_k V_k^\rmT$ be the compact SVD of $\phi_k$ with the singular values on the diagonal of $\Sigma_k$ in descending order.\footnote{If $p \le n$, then $U_k \in \BBR^{p \times p}$, $\Sigma_k \in \BBR^{p \times p}$, and $V_k \in \BBR^{n \times p}$. If $p > n$, then $U_k \in \BBR^{p \times n}$, $\Sigma_k \in \BBR^{n \times n}$, and $V_k \in \BBR^{n \times n}$.}\footnote{This is more computationally efficient than computing the full SVD of $\phi_k$. The compact SVD can, for example, be computed in MATLAB as \texttt{[U,S,V] = svd(phi,"econ")}.}\footnote{Note that, for computational efficiency, the right singular vectors $V_k$ do not need to be computed.} 
Let $q_k$ be the number of singular values of $\phi_k$ greater or equal to $\sqrt{\varepsilon}$ and define $\bar{U}_k \in \BBR^{p \times q_k}$ as the first $q_k$ columns of $U_k$. 
Then, define the filtered regressor, $\bar{\phi}_k \in \BBR^{q_k \times n}$, and filtered measurement $\bar{y}_k \in \BBR^{q_k}$ by
\begin{align}
    \bar{\phi}_k \triangleq \bar{U}_k^\rmT \phi_k, \label{eqn: phibar defn} \\
    \bar{y}_k \triangleq \bar{U}_k^\rmT y_k. \label{eqn: ybar defn}
\end{align}
Note that, by construction, all singular values of $\bar{\phi}_k$ are greater or equal to $\sqrt{\varepsilon}$.

\subsection{Subspace of Information Forgetting (SIFting)}
For all $k \ge 0$, define $R_k^{\parallel} \in \BBR^{n \times n}$ and $R_k^{\perp} \in \BBR^{n \times n}$ as
\begin{align}
    R_k^{\parallel} &\triangleq R_k^{\parallel \mathcal{R}(\bar{\phi}_k^\rmT)}
    = R_k \bar{\phi}_k^\rmT (\bar{\phi}_k R_k \bar{\phi}_k^\rmT)^{-1} \bar{\phi}_k R_k,
    \\
    R_k^{\perp} &\triangleq R_k^{\perp \mathcal{R}(\bar{\phi}_k^\rmT)}
    = R_k - R_k \bar{\phi}_k^\rmT (\bar{\phi}_k R_k \bar{\phi}_k^\rmT)^{-1} \bar{\phi}_k R_k,
\end{align}
where the operators $\parallel$ and $\perp$ are defined in Definition \ref{defin: subspace decomposition} and $\mathcal{R}(\bar{\phi}_k^\rmT)$ is the information subspace.
It will later be shown in Corollary \ref{cor: phi R phi^T well conditioned} that the matrix $\bar{\phi}_k R_k \bar{\phi}_k^\rmT$ is nonsingular and, under mild conditions, is well-conditioned.
Hence, for all $k \ge 0$, the information subspace decomposition of $R_k$ can be expressed as
\begin{align}
    \label{eqn: Rk decomp}
    R_k = R_k^{\parallel} + R_k^{\perp}.
\end{align}
Next, we forget in only the information subspace to get the \textit{SIFted information matrix} $\bar{R}_k \in \BBR^{n \times n}$, defined as
\begin{align}
    \bar{R}_k &\triangleq \lambda R_k^\parallel  + R_k^\perp \nonumber \\
    & = R_k - (1-\lambda) R_k \bar{\phi}_k^\rmT (\bar{\phi}_k R_k \bar{\phi}_k^\rmT)^{-1} \bar{\phi}_k R_k.
    \label{eqn: Rkbar defn}
\end{align}
Furthermore, we define the \textit{SIFted covariance matrix} as $\bar{P}_k \triangleq \bar{R}_k^{-1} \in \BBR^{n \times n}$.
Note that it will later be shown in Corollary \ref{cor: Rk pos def} that $\bar{R}_k$ is positive definite, and hence nonsingular.
The matrix inversion lemma (Lemma \ref{lem: matrix inversion lemma}) can also be used to express $\bar{P}_k$ as 
\begin{align}
    \label{eqn: SIFt Pkbar update}
    \bar{P}_k &= P_k + \frac{1-\lambda}{\lambda} \bar{\phi}_k^\rmT (\bar{\phi}_k R_k \bar{\phi}_k^\rmT )^{-1} \bar{\phi}_k.
\end{align}
%
%
\subsection{Update}

Lastly, for all $k \ge 0$, we compute the updated information matrix, $R_{k+1} \in \BBR^{n \times n}$, and updated parameter vector estimate, $\theta_{k+1} \in \BBR^n$ as
\begin{align}
    R_{k+1} &= \bar{R}_k + \bar{\phi}_k^\rmT \bar{\phi}_k, \label{eqn: SIFt Rk Update}
    \\
    \theta_{k+1} &= \theta_k + P_{k+1} \bar{\phi}_k^\rmT(\bar{y}_k-\bar{\phi}_k \theta_k), \label{eqn: SIFt theta_k update}
\end{align}
where $P_{k+1} \triangleq R_{k+1}^{-1} \in \BBR^{n \times n}$. It will later be shown in Corollary \ref{cor: Rk pos def} that $R_{k+1}$ is positive definite, and hence nonsingular.
Furthermore, note that the matrix inversion Lemma (Lemma \ref{lem: matrix inversion lemma}) can be used to express $P_{k+1}$ as
\begin{align}
    \label{eqn: SIFt Pk update}
    P_{k+1} &= \bar{P}_k - \bar{P}_k \bar{\phi}_k^\rmT (I_{q_k} + \bar{\phi}_k \bar{P}_k \bar{\phi}_k^\rmT)^{-1} \bar{\phi}_k \bar{P}_k.
\end{align}


\subsection{Implementation and Computational Cost}
\label{subsec: SIFt computational cost}

This subsection summarizes some subtleties in implementing SIFt-RLS which minimize computational cost and ensure numerical stability. 
We also analyze the computational complexity of SIFT-RLS and discuss when it is or is not efficient to use the matrix inversion Lemma to update the covariance matrix.

To begin our discussion on the computational cost of SIFt-RLS, note that, for all $k \ge 0$, the computational complexity of the compact SVD of $\phi_k$ is $\mathcal{O}(\min\{pn^2,p^2n\})$ (See Lecture 31 of \cite{trefethen2022numerical}).
Next, note that $\bar{R}_k$, defined in \eqref{eqn: Rkbar defn}, can be computed in $\mathcal{O}(q_k n^2)$ time complexity, where $q_k \le \min\{p,n\}$ be the number of singular values of $\phi_k$ larger than $\sqrt{\varepsilon}$. 
This is done by first computing $L_k \triangleq \bar{\phi}_k R_k$, then computing $\bar{R}_k = R_k - (1-\lambda)L_k^\rmT(L_k \bar{\phi}_k^\rmT)^{-1} L_k$.\footnote{Computing $L_k$ is important because standard left to right multiplication of \eqref{eqn: Rkbar defn} results in $\mathcal{O}(n^3)$ time complexity.}
A concern in practice is that the calculation of $\bar{R}_k$ over many steps may cause $\bar{R}_k$ to drift from symmetric due to round-off errors, resulting in numerical instability.
A simple solution is to recompute $\bar{R}_k \gets \frac{1}{2}(\bar{R}_k + \bar{R}_k^\rmT)$ at each step.


Furthermore, note that, for all $k \ge 0$, the updated parameter estimate $\theta_{k+1}$, given by \eqref{eqn: SIFt theta_k update} is most efficiently computed with the ordering $\theta_{k+1} = \theta_k + P_{k+1} [\bar{\phi}_k^\rmT(\bar{y}_k-\bar{\phi}_k \theta_k)]$. 
This computation requires the updated covariance matrix $P_{k+1}$, which can either be computed by first computing $\bar{P}_k$ using \eqref{eqn: SIFt Pkbar update} then computing $P_{k+1}$ using \eqref{eqn: SIFt Pk update}, or by directly inverting $R_{k+1}$.

Using the matrix inversion lemma, the matrix $\bar{P}_k$ can be computed in $\mathcal{O}(q_k n^2)$ as $\bar{P}_k = P_k + \frac{1-\lambda}{\lambda} \bar{\phi}_k^\rmT (L_k \bar{\phi}_k^\rmT )^{-1} \bar{\phi}_k$.
Then, $P_{k+1}$ can be computed in $\mathcal{O}(q_k n^2)$ by first computing $M_k \triangleq \bar{\phi}_k \bar{P}_k$, then computing $P_{k+1} = \bar{P}_k - M_k^\rmT (I_{q_k} + M_k \bar{\phi}_k^\rmT)^{-1} M_k$.\footnote{First computing $M_k$ is important because standard left to right matrix multiplication of \eqref{eqn: SIFt Pk update} results in $\mathcal{O}(n^3)$ time complexity.}
Similarly, $P_{k+1}$ may drift from symmetric when computed using the matrix inversion Lemma due to round-off errors. 
A solution is to recompute $P_{k+1} \gets \frac{1}{2}(P_{k+1} + P_{k+1}^\rmT)$ at each step.
In contrast, the computational complexity to compute $P_{k+1}$ by directly inverting $R_{k+1}$ is $\mathcal{O}(n^3)$.

In practice, however, it is only faster to compute $P_{k+1}$
using the matrix inversion lemma via \eqref{eqn: SIFt Pkbar update} and \eqref{eqn: SIFt Pk update} if $q_k \ll n$.
Note that $q_k \in \{0,1,\hdots, \min\{p,n\} \}$ may change at each step but $n$ remains constant.
Hence, we suggest to select $q_{\rm max} \in \{1,2,\hdots,\min\{p,n\}\}$ such that if $q_k \le q_{\rm max}$, $P_{k+1}$ is computed using \eqref{eqn: SIFt Pkbar update} and \eqref{eqn: SIFt Pk update}, otherwise $P_{k+1}$ is computed by directly inverting $R_{k+1}$.
Figure \ref{fig:SIFt MIL or not} gives data on which method is faster for various values of $q_k$ and $n$ as well as suggest values of $q_{\rm max}$ as a function of $n$.
It is important to note that the choice of $q_{\rm max}$ only affects computational cost, not the algorithm itself. 
Alternatively, for simpler implementation, one can select $q_{\rm max} = 0$ to always perform direct inversion or $q_{\rm max} = p$ to always use the matrix inversion lemma.

\begin{figure}[ht]
    \centering
    \includegraphics[width = .47 \textwidth]{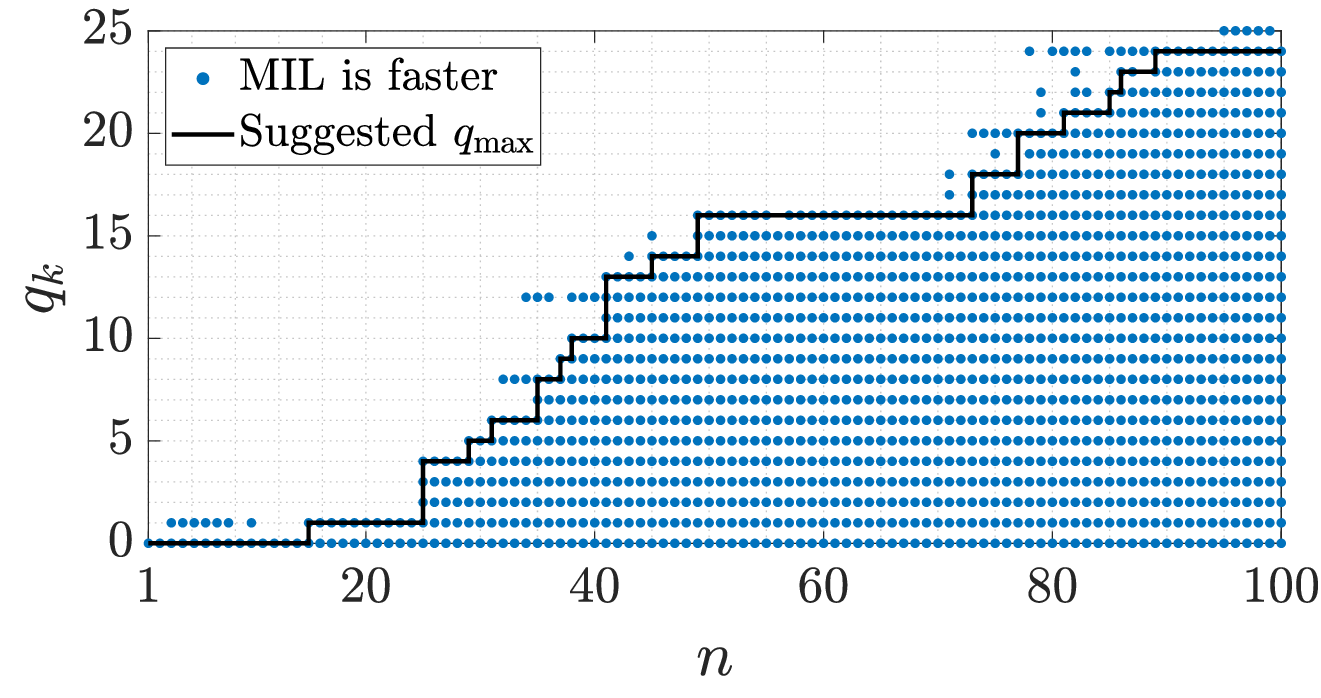}
    \caption{Comparison of computation time to compute $P_{k+1}$ in SIFt-RLS, tested in MATLAB on an i7-6600U processor with 16 GB of RAM. 
    Blue x indicates that, for the given values of $q_k$ and $n$, it is faster, on average, to compute $P_{k+1} \in \BBR^{n \times n}$ using the matrix inversion lemma via \eqref{eqn: SIFt Pkbar update} and \eqref{eqn: SIFt Pk update} than by direct inversion of $R_{k+1}$. 
    However, further testing has shown that these results may differ greatly between different machines.}
    \label{fig:SIFt MIL or not}
\end{figure}

To summarize, the algorithm SIFt-RLS can be performed in $\max\{\min\{pn^2,p^2n\},q_k n^2\}$ computational complexity per step.
While no assumptions are made whether $p < n$ or $p \ge n$, SIFt-RLS is best used when $p < n$.
In the case where $p < n$, the computational complexity of SIFt-RLS is $\max\{p^2 n , q_k n^2\}$ per step.
Note that if $p \ge n$ and, at step $k \ge 0$, all singular values of $\phi_k$ are greater or equal to $\sqrt{\varepsilon}$, then the information subspace is all of $\BBR^n$ and SIFt-RLS simplifies to exponential forgetting RLS, as shown in Proposition \ref{prop: SIFt becomes EF RLS}.
However, SIFt-RLS can still be still be useful in the case $p \ge n$ if it is expected that $\phi_k$ will have one or more singular values smaller than $\sqrt{\varepsilon}$.

\begin{prop}
    \label{prop: SIFt becomes EF RLS}
    Let $k \ge 0$, let $\phi_k \in \BBR^{p \times n}$, $y_k \in \BBR^p$, $\theta_k \in \BBR^n$, and let $R_k \in \BBR^{n \times n}$ be positive definite. 
    Also let $\varepsilon > 0$.
    Let $\theta_{k+1} \in \BBR^n$ and $R_{k+1} \in \BBR^{n \times n}$ be computed using \eqref{eqn: phibar defn}, \eqref{eqn: ybar defn}, \eqref{eqn: Rkbar defn}, \eqref{eqn: SIFt Rk Update}, and \eqref{eqn: SIFt theta_k update}.
    If $p \ge n$ and all singular values of $\phi_k$ are greater or equal to $\sqrt{\varepsilon}$, then
    \begin{align}
        R_{k+1} &= \lambda R_k + \phi_k^\rmT \phi_k, 
        \label{eqn: SIFt equivalent to EF Rk update}
        \\
        \theta_{k+1} &= \theta_k + P_{k+1} \phi_k^\rmT(y_k-\phi_k \theta_k).
        \label{eqn: SIFt equivalent to EF thetak update}
    \end{align}
\end{prop}
\begin{proof}
Let $\phi_k = U_k \Sigma_k V_k^\rmT$ be the compact SVD of $\phi_k$, where $U_k \in \BBR^{p \times n}$, $\Sigma_k \in \BBR^{n \times n}$, and $V_k \in \BBR^{n \times n}$.  
Since $p \ge n$ and all singular values of $\phi_k$ are greater or equal to $\sqrt{\varepsilon}$, it follows that $q_k = n$ and $\bar{U}_k = U_k$.
Hence, $\bar{\phi}_k = U_k^\rmT \phi_k \in \BBR^{n \times n}$ and $\bar{y}_k = U_k^\rmT y_k \in \BBR^n$. 
Moreover, since $U_k$ is semi-orthogonal, $\rank(\bar{\phi}_k) = n$, hence $\bar{\phi}_k$ is nonsingular.
Therefore, it follows that $(\bar{\phi}_k R_k \bar{\phi}_k^\rmT)^{-1} = \bar{\phi}_k^{-\rmT} P_k \bar{\phi}_k^{-1}$.
Substituting this equality into \eqref{eqn: Rkbar defn}, it follows that $\bar{R}_k = \lambda R_k$.
Furthermore, since $U_k$ is semi-orthogonal, it follows that $\bar{\phi}_k^\rmT \bar{\phi}_k = \phi_k^\rmT \phi_k$ and $\bar{\phi}_k^\rmT \bar{y}_k = \phi_k^\rmT y_k$.
Substituting into \eqref{eqn: SIFt Rk Update} and \eqref{eqn: SIFt theta_k update}, \eqref{eqn: SIFt equivalent to EF Rk update} and \eqref{eqn: SIFt equivalent to EF thetak update} follow.
\end{proof}

The implementation details discussed in this subsection are summarized in Algorithm \ref{alg:SIFt}.
Some recommendations for tuning parameters $\lambda \in (0,1)$, $\varepsilon > 0$, $\theta_0 \in \BBR^n$, and positive-definite $P_0 \in \BBR^{n \times n}$ are as follows.
$\theta_0$ and $P_0$ are a prior belief of the mean and covariance, respectively, of the parameters being estimated and can be tuned similarly as they would in other extensions of RLS. 
$\varepsilon$ can be determined based on the noise level of the data but should be small to avoid overly truncating the data.
Finally, while a forgetting factor close to $1$ is traditionally recommended to avoid covariance windup \cite{bruce2020convergence,paleologu2014practical,paleologu2008robust}, we have found $\lambda \in (0,1)$ can be tuned more aggressively (more small) in SIFt-RLS to track quickly changing parameters since forgetting only occurs in the information subspace. 
This is particularly useful if there is ample excitation in certain directions, allowing for identification of a subspace of parameters, while there is little excitation in other directions. 

\algnewcommand\algorithmicinput{\textbf{Tuning Parameters:}}
\algnewcommand\Tuning{\item[\algorithmicinput]}
\algnewcommand\algorithmicinputt{\textbf{Computational Parameters:}}
\algnewcommand\Numerical{\item[\algorithmicinputt]}

\begin{algorithm}[ht]
\caption{Subspace of Information Forgetting Recursive Least Squares (SIFt-RLS)}\label{alg:SIFt}
\begin{algorithmic}[1]

\Tuning $\lambda \in (0,1)$, $\varepsilon > 0$, $\theta_0 \in \BBR^n$, positive-definite $R_0 \in \BBR^{n \times n}$ and $P_0 \triangleq R_0^{-1}$.
\Numerical $q_{\rm max} \in [1,\min\{p,n\}]$.
\ForAll{$k \ge 0$}
    \State \uline{\textbf{1. Information Filtering:}} \vspace{2pt}
    \State Compute the compact SVD $\phi_k = U_k \Sigma_k V_k^\rmT$ with the singular values on the diagonal of $\Sigma_k$ in descending order.
            Let $q_k$ be the number of singular values of $\phi_k$ greater or equal to $\sqrt{\varepsilon}$. 
    \If{$q_k = 0$}        
    \State $R_{k+1} \gets R_k$, $\theta_{k+1} \gets \theta_k$, skip to step $k+1$.
    \EndIf
    \State Let $\bar{U}_k$ be the first $q_k$ columns of $U_k$.
    \State $\bar{\phi}_k \in \BBR^{q_k \times n} \gets \bar{U}_k^\rmT \phi_k$ 
    \State $\bar{y}_k \in \BBR^{q_k} \gets \bar{U}_k^\rmT y_k$ 
    \State \uline{\textbf{2. SIFting:}} \vspace{2pt}
    \State $L_k \gets \bar{\phi}_k R_k$
    \State $\bar{R}_k \gets R_k - (1-\lambda) L_k^\rmT (L_k \bar{\phi}_k^\rmT)^{-1} L_k$
    \State $\bar{R}_k \gets \frac{1}{2}(\bar{R}_k + \bar{R}_k^\rmT)$
    \If{$q_k \le q_{\rm max}$} 
    \State $\bar{P}_k \gets P_k + \frac{1-\lambda}{\lambda} \bar{\phi}_k^\rmT (L_k \bar{\phi}_k^\rmT )^{-1} \bar{\phi}_k.$
    \EndIf
    \State \uline{\textbf{3. Update:}} \vspace{2pt}
    \State $R_{k+1} \gets \bar{R}_k + \bar{\phi}_k^\rmT \bar{\phi}_k$
    \If{$q_k \le q_{\rm max}$} 
    \State $M_k \gets \bar{\phi}_k \bar{P}_k$ 
    \State $P_{k+1} \gets \bar{P}_k - M_k^\rmT (I_{q_k} + M_k \bar{\phi}_k^\rmT)^{-1} M_k$
    \State $P_{k+1} \gets \frac{1}{2}(P_{k+1} + P_{k+1}^\rmT)$
    \Else
    \State $P_{k+1} \gets R_{k+1}^{-1}$
    \EndIf
    \State $\theta_{k+1} \gets \theta_k + P_{k+1} [\bar{\phi}_k^\rmT(\bar{y}_k-\bar{\phi}_k \theta_k)]$
\EndFor
\end{algorithmic}
\end{algorithm}

\section{Covariance Bounds and Stability of SIFt-RLS}

This section discusses the theoretical guarantees of SIFt-RLS including bounds on the covariance matrix and estimation error stability assuming estimation of fixed parameters without noise.

\subsection{Covariance Bounds}

To begin, we show that, without any assumptions of persistent excitation, the eigenvalues of the information matrix $R_k$ in SIFt-RLS are lower bounded. 
Furthermore, we show that if the sequence of regressors $(\phi_k)_{k=0}^\infty$ is upper bounded, then the eigenvalues of the information matrix $R_k$ in SIFt-RLS are also upper bounded.
This result also immediately implies that the reciprocal expressions are bounds for the covariance matrix $P_k = R_k^{-1}$. 
These results are given in Theorem \ref{theo: Rk upper and lower bounds}.
Note that explicit bounds are given which improves upon the works of \cite{cao2000directional} and \cite{zhu2021recursive} which only state the existence of bounds.
Moreover, the bounds of Theorem \ref{theo: Rk upper and lower bounds} are stronger than those of \cite{zhu2021recursive}, see Appendix \ref{sec: issues with oblique} for details.

\begin{theo}
\label{theo: Rk upper and lower bounds}
Let $\lambda \in (0,1)$, let $\varepsilon > 0$, and let $R_0 \in \BBR^{n \times n}$ be positive definite. 
For all $k \ge 0$, let $\phi_k \in \BBR^{p \times n}$, let $\bar{\phi}_k \in \BBR^{q_k \times n}$ be defined by \eqref{eqn: phibar defn}, let $\bar{R}_k \in \BBR^{n \times n}$ be defined by \eqref{eqn: Rkbar defn}, and let $R_{k+1} \in \BBR^{n \times n}$ be defined by \eqref{eqn: SIFt Rk Update}. 
Then the following statements hold:
\begin{enumerate}
    \item For all $k \ge 0$,
    \begin{align}
        \eigmin(R_k) &\ge \min \{ \frac{\varepsilon}{1-\lambda}, \eigmin(R_0) \}, \label{eqn: Rk lower bound} 
        \\
        \eigmax(P_k) &\le \max \{ \frac{1-\lambda}{\varepsilon}, \eigmax(P_0) \}. \label{eqn: P_k upper bound}
    \end{align}
    \item If $(\phi_k)_{k=0}^\infty$ is bounded with upper bound $\beta \in (0,\infty)$, then, for all $k \ge 0$,
    \begin{align}
        \eigmax(R_k) &\le \max \{ \frac{\beta}{1-\lambda}, \eigmax(R_0) \} , \label{eqn: Rk upper bound}
        \\
        \eigmin(P_k) &\ge \min \{ \frac{1-\lambda}{\beta}, \eigmin(P_0) \}. \label{eqn: P_k lower bound}
    \end{align}
\end{enumerate}
\end{theo}
\begin{proof}
    See Appendix section \ref{sec: proof of covariance bounds} for a proof of \eqref{eqn: Rk lower bound} and \eqref{eqn: Rk upper bound}. Bounds \eqref{eqn: P_k upper bound} and \eqref{eqn: P_k lower bound} follow directly from \eqref{eqn: Rk lower bound} and \eqref{eqn: Rk upper bound} respectively since, for all $k \ge 0$, $P_k = R_k^{-1}$.
\end{proof}

Note that the bounds of Theorem \ref{theo: Rk upper and lower bounds} are strikingly analogous to the bounds of RLS with exponential forgetting in Proposition \ref{prop: EF RLS Covariance Bound with N=1}, where $(\phi_k)_{k=0}^\infty$ is persistently exciting with persistency window $N=1$.
This is a result of the information filtering step of SIFt-RLS truncating the singular values of $\phi_k$ smaller than $\sqrt{\varepsilon}$, and the SIFting step of SIFt-RLS forgetting only in the information subspace, $\mathcal{R}({\phi}_k^\rmT)$.

Next, Corollary \ref{cor: Rk pos def} shows that, for all $k \ge 0$, $R_k$ and $\bar{R}_k$ are positive definite. 
Furthermore, Corollary \ref{cor: phi R phi^T well conditioned} shows that, for all $k \ge 0$, the matrix $\bar{\phi}_k R_k \bar{\phi}_k^\rmT$ is nonsingular and, assuming bounded regressors, has a bounded condition number. 
This is ensures that computing $(\bar{\phi}_k R_k \bar{\phi}_k^\rmT)^{-1}$ in the SIFting step is well-conditioned.

\begin{cor}
\label{cor: Rk pos def}
    Consider the notation of Theorem \ref{theo: Rk upper and lower bounds}.
    For all $k \ge 0$, $R_k$ and $\bar{R}_k$ are positive definite.
\end{cor}
\begin{proof}
    Since $\min\{\frac{\varepsilon}{1-\lambda},\eigmin(R_0)\} > 0$, it follows from \eqref{eqn: Rk lower bound} that, for all $k \ge 0$, $R_k \succ 0$.
    Next, note that, for all $k \ge 0$, $\bar{R}_k$, defined in \eqref{eqn: Rkbar defn}, can be written as $\bar{R}_k = \lambda R_k + (1-\lambda) R_k^\perp$.
    It follows from Theorem \ref{theo: Apar Aperp properties} that $R_k^\perp$ is positive semidefinite, and hence $\bar{R}_k \succeq \lambda R_k \succ 0$.
\end{proof}

\begin{cor}
\label{cor: phi R phi^T well conditioned}
    Consider the notation of Theorem \ref{theo: Rk upper and lower bounds}.
    For all $k \ge 0$, $\bar{\phi}_k R_k \bar{\phi}_k^\rmT$ is nonsingular. 
    Moreover, if $(\phi_k)_{k=0}^\infty$ is bounded with upper bound $\beta \in (0,\infty)$, then, for all $k \ge 0$,
    \begin{align}
    \label{eqn: phi R phi^T well conditioned}
        \kappa(\bar{\phi}_k R_k \bar{\phi}_k^\rmT) \le \frac{\beta \max \{ \frac{\beta}{1-\lambda}, \eigmax(R_0) \}}{\varepsilon \min \{ \frac{\varepsilon}{1-\lambda}, \eigmin(R_0) \} }.
    \end{align}
\end{cor}
\begin{proof}
    It follows from Corollary \ref{cor: Rk pos def} that $R_k$ is positive definite.
    Next, since $\rank(\bar{\phi}_k) = q_k$ by construction, it follows from Lemma \ref{lem: phi R phi^T PSD} that $\bar{\phi}_k R_k \bar{\phi}_k^\rmT$ is positive definite, hence nonsingular. 
    
    Next, suppose $(\phi_k)_{k=0}^\infty$ is bounded with upper bound $\beta \in (0,\infty)$.
    Since $\svdmax(\cdot)$ is submultiplicative, it follows that $\svdmax(\bar{\phi}_k R_k \bar{\phi}_k^\rmT) \le \svdmax(\bar{\phi}_k) \svdmax(R_k) \svdmax(\bar{\phi}_k^\rmT)$.
    Furthermore, since $(\phi_k)_{k=0}^\infty$ is bounded with upper bound $\beta \in (0,\infty)$, $(\bar{\phi}_k)_{k=0}^\infty$ is also bounded with upper bound $\beta \in (0,\infty)$.
    Therefore, $\svdmax(\bar{\phi}_k) \le \sqrt{\beta}$.
    This combined with \eqref{eqn: Rk upper bound} imply that $\svdmax(\bar{\phi}_k R_k \bar{\phi}_k^\rmT) \le \beta \max \{ \frac{\beta}{1-\lambda}, \eigmax(R_0) \}$.

    Next, as a consequence of the min-max theorem given in Lemma \ref{lem: min-max theorem}, note that 
    \begin{align*}
        \svdmin(\bar{\phi}_k R_k \bar{\phi}_k^\rmT) = \min_{x \in \BBR^{q_k}, x\ne 0} \frac{\Vert \bar{\phi}_k R_k \bar{\phi}_k^\rmT x \Vert}{\Vert x \Vert}.
    \end{align*}
    Since $\bar{\phi}_k^\rmT$ has full column rank, $x \ne 0$ implies that $\bar{\phi}_k^\rmT x \ne 0$. Therefore, we can write
    \begin{align*}
        \min_{x \in \BBR^{q_k}, x\ne 0} & \frac{\Vert \bar{\phi}_k R_k \bar{\phi}_k^\rmT x \Vert}{\Vert x \Vert}
        = \min_{x \in \BBR^{q_k} , x\ne 0} \frac{\Vert \bar{\phi}_k R_k \bar{\phi}_k^\rmT x \Vert}{\Vert \bar{\phi}_k^\rmT x \Vert } \frac{\Vert \bar{\phi}_k^\rmT x \Vert}{\Vert x \Vert}
        \\
        & \ge \min_{y \in \BBR^n,y\ne 0} \frac{\Vert \bar{\phi}_k R_k y \Vert}{\Vert y \Vert }
        \min_{x \in \BBR^{q_k} , x\ne 0} \frac{\Vert \bar{\phi}_k^\rmT x \Vert}{\Vert x \Vert} 
        \\
        & = \svdmin(\bar{\phi}_k R_k) \svdmin(\bar{\phi}_k^\rmT).
    \end{align*}
    Furthermore, since $\bar{\phi}_k$ has full row rank, it can be shown by similar reasoning that $\svdmin(\bar{\phi}_k R_k) \ge \svdmin(\bar{\phi}_k) \svdmin(R_k)$.
    By construction of $\bar{\phi}_k$, $\svdmin(\bar{\phi}_k) = \svdmin(\bar{\phi}_k^\rmT) \ge \sqrt{\varepsilon}$.
    Combined with \eqref{eqn: Rk lower bound}, it follows that $\svdmin(\bar{\phi}_k R_k \bar{\phi}_k^\rmT) \ge \varepsilon \svdmin(R_k) \ge \varepsilon \min \{ \frac{\varepsilon}{1-\lambda}, \eigmin(R_0) \}$.

    Substituting the upper bound of $\svdmax(\bar{\phi}_k R_k \bar{\phi}_k^\rmT)$ and lower bound of $\svdmin(\bar{\phi}_k R_k \bar{\phi}_k^\rmT)$ into definition \eqref{eq::Cond_numb_Defn} yields \eqref{eqn: phi R phi^T well conditioned}.
\end{proof}

\subsection{Stability}

Next, for the analysis of this section, we make the assumption that there exist fixed parameters $\theta \in \BBR^n$ such that, for all $k \ge 0$, 
\begin{align}
    \label{eqn: y = phi theta}
    y_k = \phi_k \theta.
\end{align}
Furthermore, for all $k \ge 0$, we define the parameter estimation error $\tilde{\theta}_k \in \BBR^n$ by
\begin{align}
    \label{eqn: thetatilde defn}
    \tilde{\theta}_k \triangleq \theta_k - \theta.
\end{align}
Substituting \eqref{eqn: thetatilde defn} into \eqref{eqn: SIFt theta_k update}, it then follows that, for all $k \ge 0$,
\begin{align}
\label{eqn: theta tilde update}
    \tilde{\theta}_{k+1} = (I_n - P_{k+1} \bar{\phi}_k^\rmT \bar{\phi}_k) \tilde{\theta}_k.
\end{align}
Note that \eqref{eqn: theta tilde update} is a linear time-varying system with an equilibrium point $\tilde{\theta}_k \equiv 0$.
Thus, we can study the stability of the estimation error equilibrium point $\tilde{\theta}_k \equiv 0$.
Theorem \ref{theo: SIFt stability} uses the results of \cite{lai2023generalized} to show that this the equilibrium point $\tilde{\theta}_k \equiv 0$ is uniformly asymptotically stable without further assumptions and globally uniformly exponentially stable under persistent excitation. 
For definitions of uniform asymptotic stability and global uniform exponential stability, see Definition 13.7 in \cite[pp. 783, 784]{Haddad2008}.

\begin{theo}
\label{theo: SIFt stability}
Let $\lambda \in (0,1)$, let $\varepsilon > 0$, let $R_0 \in \BBR^{n \times n}$ be positive definite, and let $\theta_0 \in \BBR^n$. 
For all $k \ge 0$, let $\phi_k \in \BBR^{p \times n}$, let $\bar{\phi}_k \in \BBR^{q_k \times n}$ be defined by \eqref{eqn: phibar defn}, let $\bar{R}_k \in \BBR^{n \times n}$ be defined by \eqref{eqn: Rkbar defn}, let $R_{k+1} \in \BBR^{n \times n}$ be defined by \eqref{eqn: SIFt Rk Update}, and let $\theta_{k+1} \in \BBR^n$ be defined by \eqref{eqn: SIFt theta_k update}.
Finally, assume there exists $\theta \in \BBR^n$ such that, for all $k \ge 0$, \eqref{eqn: y = phi theta} holds.
Then, the following two statements hold:
    \begin{enumerate}
        \item The equilibrium $\tilde{\theta}_k \equiv 0$ of \eqref{eqn: theta tilde update} is uniformly Lyapunov stable.
        \item If $\{\bar{\phi}_k\}_{k=0}^\infty$ is persistently exciting, then the equilibrium $\tilde{\theta}_k \equiv 0$ of \eqref{eqn: theta tilde update} is globally uniformly exponentially stable.
    \end{enumerate}
\end{theo}

\begin{proof}
    This result is based on Theorem 2 of \cite{lai2023generalized}.
    Note that, for all $k \ge 0$, \eqref{eqn: SIFt Rk Update} can be written as 
    \begin{align*}
        R_{k+1} = R_k - F_k + \bar{\phi}_k^\rmT \bar{\phi}_k,
    \end{align*}
    where $F_k \triangleq (1-\lambda)R_k^\parallel$. Note that, for all $k \ge 0$, $1-\lambda > 0$ and that by Theorem \ref{theo: Apar Aperp properties}, $R_k^\parallel \succeq 0$, hence $F_k \succeq 0$. Therefore, condition \textit{A1)} of \cite{lai2023generalized} is satisfied. 
    
    Next, note that, for all $k \ge 0$, $R_k - F_k = \lambda R_k^\parallel + R_k^\perp$. Then, since $R_k^\perp \succeq 0$ by Theorem \ref{theo: Apar Aperp properties}, it follows that
    \begin{align*}
        R_k - F_k \succeq \lambda R_k^\parallel + \lambda R_k^\perp = \lambda R_k.
    \end{align*}
    Furthermore, it follows from \eqref{eqn: Rk lower bound} of Theorem \ref{theo: Rk upper and lower bounds} that, for all $k \ge 0$,
    \begin{align*}
        R_k - F_k \succeq \lambda \min \{ \frac{\varepsilon}{1-\lambda}, \eigmin(R_0) \} I_n,
    \end{align*}
    and hence
    \begin{align*}
        (R_k - F_k)^{-1} \preceq \frac{1}{\lambda} \max \{ \frac{1-\lambda}{\varepsilon}, \eigmax(P_0) \} I_n.
    \end{align*}
    Therefore, condition \textit{A2)} of \cite{lai2023generalized} is satisfied. Finally, \eqref{eqn: Rk upper bound} of Theorem \ref{theo: Rk upper and lower bounds} implies that, for all $k \ge 0$, 
    \begin{align*}
        \min \{ \frac{1-\lambda}{\beta} , \eigmin(P_0) \} I_n \preceq P_k .
    \end{align*}
    Hence, condition \textit{A2)} of \cite{lai2023generalized} is satisfied. Thus, by statement 2) of Theorem 2 of \cite{lai2023generalized}, the equilibrium $\tilde{\theta}_k \equiv 0$ of \eqref{eqn: theta tilde update} is uniformly Lyapunov stable.
    
    Moreover, if $\{\bar{\phi}_k\}_{k=0}^\infty$ is persistently exciting, then condition \textit{A4)} of \cite{lai2023generalized} is satisfied. It then follows from statement 4) of Theorem 2 of \cite{lai2023generalized}, the equilibrium $\tilde{\theta}_k \equiv 0$ of \eqref{eqn: theta tilde update} is globally uniformly exponentially stable.
\end{proof}

\section{Numerical Example}


%
We consider the follow example of identifying $n = 4$ parameters with $p = 2$ measurements at each step. For all $0 \le k \le 1200$, consider the true parameters $\theta_{{\rm true},k} \in \BBR^4$, defined
\begin{align}
    \theta_{{\rm true},k} &\triangleq 
    \begin{bmatrix}
        \theta^1_{{\rm true},k} & \theta^2_{{\rm true},k} & \theta^3_{{\rm true},k} & \theta^4_{{\rm true},k}
    \end{bmatrix}^\rmT \nonumber
    \\
    &= \begin{bmatrix}
        \sin (\frac{\pi}{60}k) & \cos (\frac{k}{60}) & \sin (\frac{k}{225}) & \cos (\frac{k}{225})
    \end{bmatrix}^\rmT.
\end{align}
For all $0 \le k < 400$, let both rows of $\phi_k \in \BBR^{2 \times 4}$ be i.i.d. sampled from $\mathcal{N}(0,\diag(1,1,10^{-4},10^{-4}))$ and, for all $400 \le k < 800$, let both rows of $\phi_k \in \BBR^{2 \times 4}$ be i.i.d. sampled from $\mathcal{N}(0,\diag(10^{-4},10^{-4},1,1))$.
Lastly, for all $800 \le k \le 1200$, let $\phi_k = \sigma_k [0 \ 2 \ 1 \ 0] + \nu_k$ where $\sigma_k$ is i.i.d. sampled from $\mathcal{N}(0,I_2)$ and both rows of $\sigma_k \in \BBR^{2 \times 4}$ are i.i.d. sampled from $\mathcal{N}(0,10^{-4} I_4)$.
Furthermore, for all $0 \le k < 1200$, let $y_k = (\phi_k + v_k)\theta_{{\rm true},k} + w_k$, where both rows of $v_k \in \BBR^{2 \times 4}$ are i.i.d. sampled from $\mathcal{N}(0,\frac{1}{100}I_4)$ and $w_k$ is i.i.d. sampled from $\mathcal{N}(0,\frac{1}{100}I_2)$, i.e. regressor noise and measurement noise.

The result of this sampled data is,  between $0 \le k < 400$, excitation is sufficient to identify $\theta^1_{{\rm true},k}$ and $\theta^2_{{\rm true},k}$, which form a basis for the excited subspace $S^{12} \triangleq \{[x_1 \ x_2 \ 0 \ 0] \colon x_1,x_2 \in \BBR\}$, but insufficient to identify $\theta^3_{{\rm true},k}$ and $\theta^4_{{\rm true},k}$, which form a basis for the unexcited subspace $(S^{12})^\perp = S^{34} \triangleq \{[0 \ 0 \ x_3 \ x_4] \colon x_3,x_4 \in \BBR\}$.
Then, the reverse is the case between $400 \le k < 800$.
Finally, between $800 \le k \le 1200$, excitation is sufficient to identify $\theta^\parallel_{{\rm true},k} \in \BBR$, defined as
\begin{equation}
    \theta^\parallel_{{\rm true},k} \triangleq \frac{2}{3} \theta^2_{{\rm true},k} + \frac{1}{3} \theta^3_{{\rm true},k},
\end{equation}
which forms a basis for the excited subspace $S^\parallel \triangleq \{[0 \ 2x \ x \ 0] \colon x \in \BBR\}$.
However, excitation is insufficient to identify  $\theta^1_{{\rm true},k}$, $\theta^4_{{\rm true},k}$, and $\theta^\perp_{{\rm true},k} \in \BBR$, defined as
\begin{align}
    \theta^\perp_{{\rm true},k} \triangleq \frac{1}{3} \theta^2_{{\rm true},k} - \frac{2}{3} \theta^3_{{\rm true},k},
\end{align}
which form a basis for the unexcited subspace $S^\perp \triangleq \{[x_1 \ x \ -2 x \ x_4] \colon x_1,x,x_4 \in \BBR\}$. 

We now consider the performance of SIFt-RLS in identifying $\theta_{{\rm true},k}$.
We initialize $P_0 = I_4$, $\theta_0 = 0_{4 \times 1}$ and let $\lambda = 0.5$ and $\varepsilon = 10^{-4}$.\footnote{We also select $q_{\rm max} = 1$ which only affects computation time.}
The parameter tracking of SIFT-RLS is shown in Figure \ref{fig: Tracking}, where, for all $k \ge 0$, $\theta^i_{k}$ is the $i^{\rm th}$ element of $\theta_k$.
Between $0 \le k < 400$, SIFt-RLS tracks $\theta^1_{{\rm true},k}$ and $\theta^2_{{\rm true},k}$ with little change to its estimate of the $\theta^3_{{\rm true},k}$ and $\theta^4_{{\rm true},k}$ and vice versa for $400 \le k < 800$.
Similarly, between $800 \le k \le 1200$, SIFt-RLS tracks $\theta^\parallel_{{\rm true},k}$ with little change to its estimate of $\theta^1_{{\rm true},k}$, $\theta^4_{{\rm true},k}$, and $\theta^\perp_{{\rm true},k}$.

\begin{figure}[ht]
    \centering
    \includegraphics[width = .47 \textwidth]{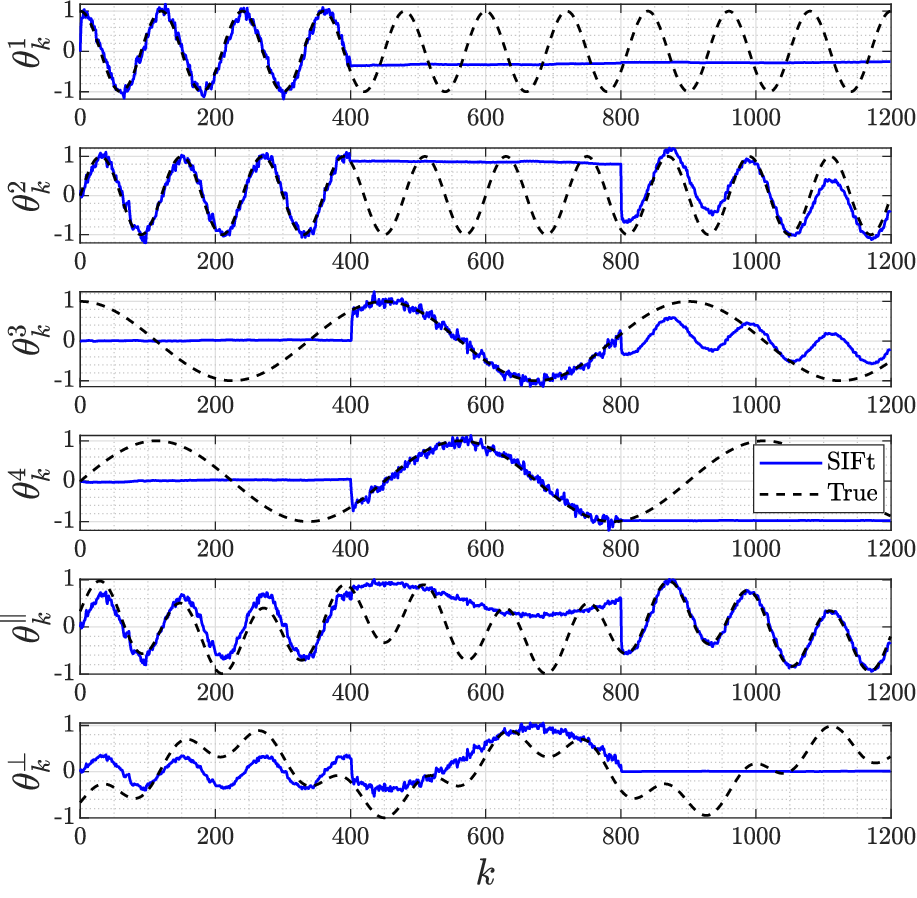}
    \caption{Estimated parameters $\theta^1_k$, $\theta^2_k$, $\theta^3_k$, $\theta^4_k$, $\theta^\parallel_k$, and $\theta^\perp_k$ using SIFt-RLS (blue) and true parameters $\theta^1_{{\rm true},k}$, $\theta^2_{{\rm true},k}$, $\theta^3_{{\rm true},k}$, $\theta^4_{{\rm true},k}$, $\theta^\parallel_{{\rm true},k}$, and $\theta^\perp_{{\rm true},k}$ (black).}
    \label{fig: Tracking}
\end{figure}

Next, we compare SIFt-RLS to RLS with no forgetting (NF), RLS with exponential forgetting (EF), and three other directional forgetting algorithms from the literature: variable direction forgetting (VD) \cite{goel2020recursive}, directional forgetting (DF) \cite{bittanti1990convergence}, and multiple forgetting (MF) \cite{vahidi2005recursive,fraccaroli2015new}.
EF and VD are tuned with a conservative $\lambda = 0.95$ to delay covariance windup. DF is tuned with the same $\lambda = 0.5$ of SIFt. MF is tuned with $\Lambda \triangleq [\lambda_1 \ \lambda_2 \ \lambda_3 \ \lambda_4] = [0.9 \ 0.9 \ 0.95 \ 0.95]$ to forget more agressively in the first and second parameters.
Finally, we test MF with a priori knowledge of which subspaces will be excited (MF a priori), where time-varying forgetting factor $\Lambda$ is tuned to only forget in the excited subspace. 
Tuning parameters are summarized in Table \ref{table:tuning}.
\begin{table}[ht]\centering
\caption{RLS Algorithms Tuning Parameters}
\setlength\tabcolsep{6pt} 
\renewcommand{\arraystretch}{1.1} 
\begin{tabular}{@{}ll@{}}
\toprule[1pt] 
Algorithm & Tuning Parameters 
\\
\midrule
SIFt & $\lambda = 0.5$, $\varepsilon = 10^{-4}$
\\
NF & --- 
\\ 
EF & $\lambda = 0.95$
\\
VD & $\lambda = 0.95$, $\varepsilon = 10^{-4}$
\\
DF & $\lambda = 0.5$
\\
MF & $\Lambda \triangleq [\lambda_1 \ \lambda_2 \ \lambda_3 \ \lambda_4] = [0.9 \ 0.9 \ 0.95 \ 0.95]$ 
\\
MF a priori & $\Lambda = \begin{cases}
    [0.5 \ 0.5 \ 1 \ 1] & 0 \le k < 400, \\
    [1 \ 1 \ 0.5 \ 0.5] & 400 \le k < 800, \\
    [1 \ 0.6 \ 0.85 \ 1] & 800 \le k \le 1200, \\
\end{cases}$
\\
\bottomrule[1pt]
\end{tabular}
\label{table:tuning}
\end{table}

To more easily compare the performance of these seven different algorithms in Figure \ref{fig: error}, we show, over $0 \le k \le 800$, the error in the $S^{12}$ subspace, $e^{12}_k \in \BBR$, and the error in the $S^{34}$ subspace, $e^{34}_k \in \BBR$, defined as
\begin{align}
    e^{12}_k &\triangleq \sqrt{(\theta^1_k - \theta^1_{{\rm true},k})^2 + (\theta^2_k - \theta^2_{{\rm true},k})^2}, \\
    e^{34}_k &\triangleq \sqrt{(\theta^3_k - \theta^3_{{\rm true},k})^2 + (\theta^4_k - \theta^4_{{\rm true},k})^2}.
\end{align}
Furthermore, Figure \ref{fig: error} also shows, between $800 \le k \le 1200$, the error in the $S^\parallel$ subspace, $e^{\parallel}_k \in \BBR$, and the error in the $S^\perp$ subspace, $e^{\perp}_k \in \BBR$, defined as
\begin{align*}
    e^{\parallel}_k &\triangleq \sqrt{(\theta^\parallel_k - \theta^\parallel_{{\rm true},k})^2}, \\
    e^{\perp}_k \hspace{-1.25pt} &\triangleq \hspace{-1pt} \sqrt{(\theta^1_k - \theta^1_{{\rm true},k})^2 + (\theta^4_k - \theta^4_{{\rm true},k})^2 + (\theta^\perp_k - \theta^\perp_{{\rm true},k})^2}.
\end{align*}
Note that NF is unable to track changing parameters and results in the largest errors in the excited subspace. 
On the contrary, while EF, VD, DF, and MF are able to track the true parameters in the excited subspaces, they all result in large error in the unexcited subspaces due to covariance windup.
DF performs the best of these methods, particularly between $0 \le k \le 800$, but still suffers from covariance windup between $800 \le k \le 1200$. 
While an aggressive $\lambda = 0.5$ is used for DF, slower windup still occurs with a more conservative $\lambda$.
Finally, MF a priori achieved very comparable performance to SIFt, with tracking in the excited subspace and little change in estimates in the unexcited subspace.
However, this required careful tuning and a priori knowledge of which subspaces would be excited.
\begin{figure*}[ht]
    \centering
    \includegraphics[width = .99 \textwidth]{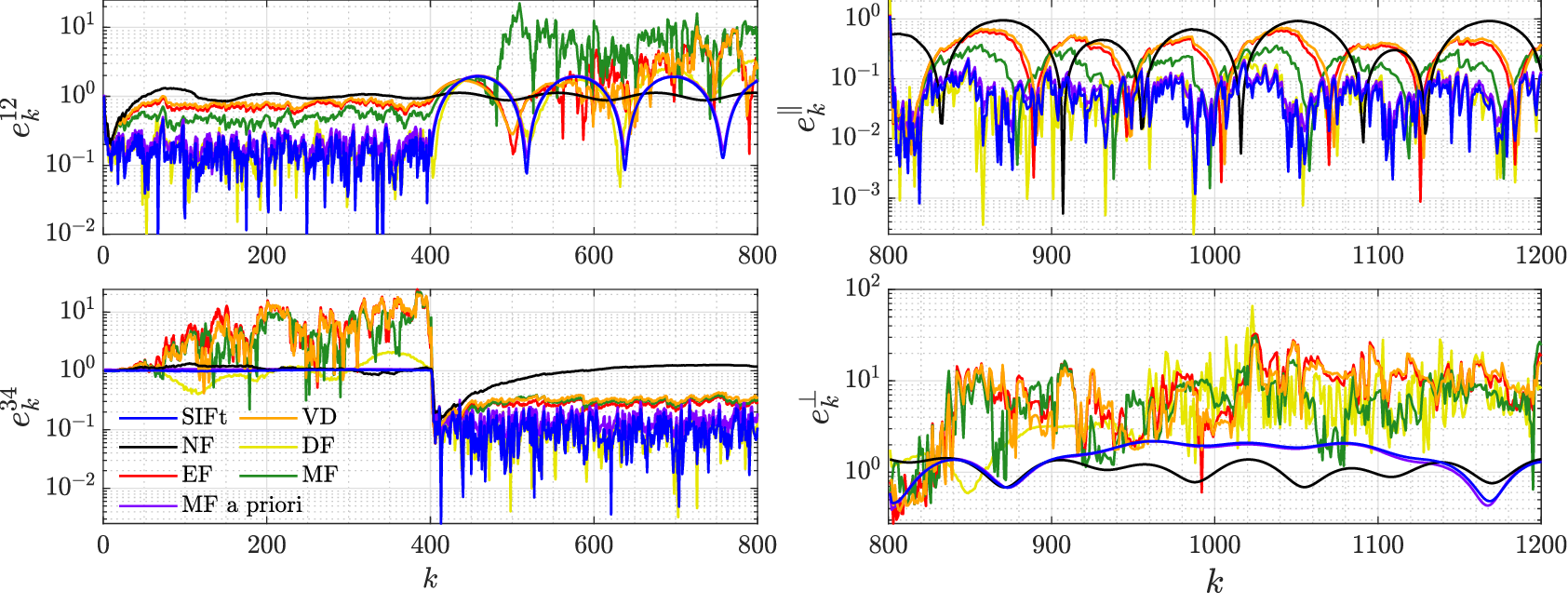}
    \caption{Over $0 \le k \le 800$, $e_k^{12}$ and $e_k^{34}$ show the parameter estimation error in the subspaces $S^{34} \triangleq \{[x_1 \ x_2 \ 0 \ 0]^\rmT \colon x_1,x_2 \in \BBR\}$ and $S_{34} \triangleq \{[0 \ 0 \ x_3 \ x_4]^\rmT \colon x_3,x_4 \in \BBR\}$, respectively. Subspace $S^{12}$ is excited over $0 \le k < 400$ while $S^{34}$ is not, and vice versa over $400 \le k < 800$.
    Next, over $800 \le k \le 1200$, $e_k^{\parallel}$ and $e_k^{\perp}$ show the parameter estimation error in the subspaces $S^\parallel \triangleq \{[0 \ 2x \ x \ 0] \colon x \in \BBR\}$ and $S^\perp \triangleq \{[x_1 \ x \ -2 x \ x_4] \colon x_1,x,x_4 \in \BBR\}$, respectively. Subspace $S^\parallel$ is excited over $800 \le k \le 1200$ while $S^\perp$ is not.
    }
    \label{fig: error}
\end{figure*}

As another useful metric, we define $\Delta^{12}_k \in \BBR$, $\Delta^{34}_k \in \BBR$, and $\Delta^{\perp}_k \in \BBR$ over $k \in [1,400]$, $k \in [401,800]$, and $k \in [801,1200]$, respectively, as
\begin{align*}
    \Delta^{12}_k &\triangleq \sqrt{(\theta^1_k - \theta^1_0)^2 + (\theta^2_k - \theta^2_{0})^2}, \\
    \Delta^{34}_k &\triangleq \sqrt{(\theta^3_k - \theta^3_{400})^2 + (\theta^4_k - \theta^4_{400})^2}, \\
    \Delta^{\perp}_k &\triangleq \sqrt{(\theta^1_k - \theta^1_{800})^2 + (\theta^4_k - \theta^4_{800})^2 + (\theta^\perp_k - \theta^\perp_{800})^2}.
\end{align*}
These quantities show how much parameter estimates change in the unexcited subspaces.
Figure \ref{fig:Delta} shows that NF, SIFt, and MF a priori makes only small changes to parameter estimates in the unexcited subspaces, while EF, VD, DF, and MF significantly update estimates of parameters in the unexcited subspace despite lack of useful new information in these directions.
\begin{figure*}[ht]
    \centering
    \includegraphics[width = .98\textwidth]{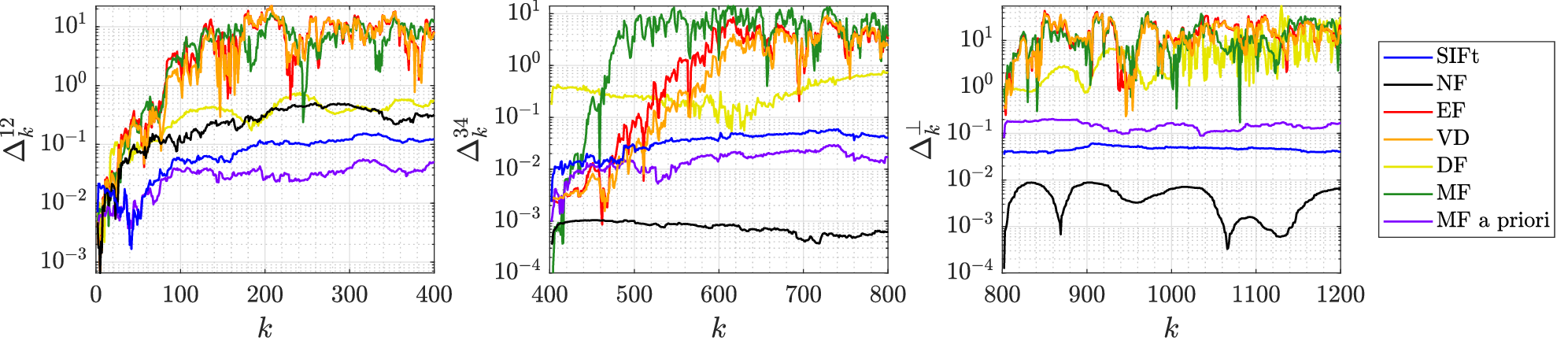}
    \caption{$\Delta_k^{12}$, $\Delta_k^{34}$, and $\Delta_k^\perp$ show the magnitude of the change in parameter estimates in the unexcited subspace over $0 \le k < 400$, $400 \le k < 800$, and $800 \le k \le 1200$, respectively.}
    \label{fig:Delta}
\end{figure*}

As a final metric, Figure \ref{fig: spectral_radius} shows the spectral radius of the covariance matrix $P_k$ for the various RLS algorithms. 
The spectral radius $\rho(P_k)$ becomes very small in NF, indicating parameter estimates is no longer changing. 
In EF, VD, DF, and MF, the spectral radius becomes very large, indicating covariance windup and sensitivity to noise.
SIFt and MF a priori both maintain a spectral radius close to $1$.

\begin{figure}[ht]
    \centering
    \includegraphics[width = .47 \textwidth]{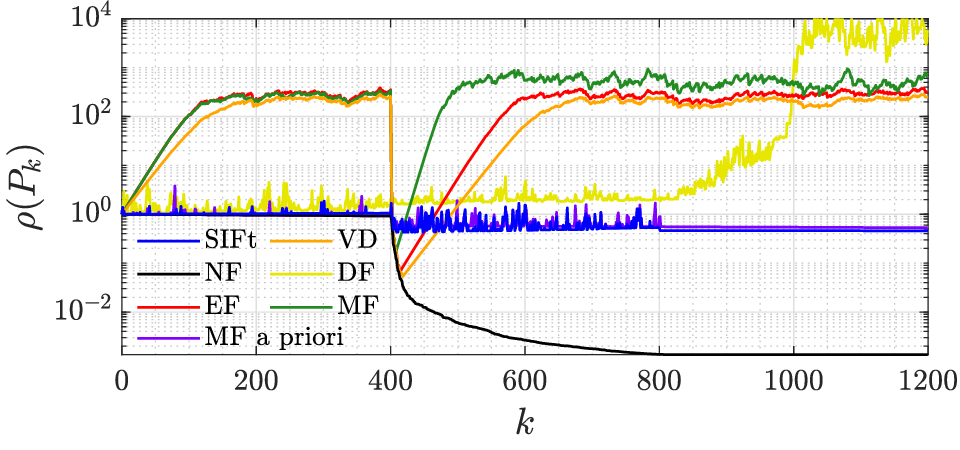}
    \caption{Spectral radius of the covariance matrix at step $k$.}
    \label{fig: spectral_radius}
\end{figure}

\section{Conclusion}

This work develops SIFt-RLS, a directional forgetting algorithm which, at each step, forgets only in the information subspace. 
The key tool used in SIFt-RLS is a subspace decomposition of a positive-definite matrix, which was developed and analyzed in Section \ref{sec: subspace decomp}.
The properties of this decomposition are used to derive explicit covariance bounds for SIFt-RLS and to develop sufficient conditions for the stability of the parameter estimation error dynamics. 
Moreover, an interesting parallel is made between the covariance bounds of SIFt-RLS and the covariance bounds of exponential forgetting RLS under persistent excitation with persistency window $1$.
A numerical example show the benefits of SIFt-RLS when the data collected is not persistently exciting.

We have presented a simple version of SIFt-RLS, in which forgetting in the information subspace is uniform and constant-rate, and no forgetting is applied in the orthogonal complement of the information subspace. 
However, a natural area of future interest is combining SIFt-RLS with other forgetting algorithms. 
One idea is to apply nonuniform forgetting and/or variable-rate forgetting to the information subspace.
Another idea is to apply a separate forgetting method, for example resetting \cite{Lai2023Exponential}, to the orthogonal complement of the information subspace.
While some other forgetting methods require persistent excitation to guarantee bounds on the covariance matrix \cite{Dasgupta1987Asymptotically,bittanti1990convergence,goel2020recursive}, a question of interest is whether SIFt-RLS combined with these methods can guarantee similar bounds without persistent excitation, as was shown in this work for exponential forgetting.

\begin{ack}                               
This work was supported by the NSF Graduate Research Fellowship under Grant DGE 1841052.  
\end{ack}

\bibliographystyle{plain}        
\bibliography{refs}           

\appendix
\section{Useful Matrix Lemmas}    

\begin{lma}[Matrix Inversion Lemma]
\label{lem: matrix inversion lemma}
Let $A \in \BBR^{n \times n}$, $U \in \BBR^{n \times p}$, $C \in \BBR^{p \times p}$, and $V \in \BBR^{p \times n}$. If $A$, $C$, and $A+UCV$ are nonsingular, then $C^{-1} + VA^{-1} U$ is nonsingular, and $(A+UCV)^{-1} = A^{-1} - A^{-1}U(C^{-1} + VA^{-1} U)^{-1} V A^{-1}$. 
\end{lma}
\begin{proof}
See Corollary 3.9.8 of \cite{bernstein2009matrix}.
\end{proof}

\begin{lma}[min-max theorem]
\label{lem: min-max theorem}
Let $A \in \BBR^{n \times n}$ be symmetric. Then, for all $k = 1,\hdots,n$,
\begin{align*}
    \boldsymbol{\lambda_k}(A) &= \max_{\substack{S \subset \BBR^n \\ \dim(S) = k}} \min_{\substack{x \in S \\ \Vert x \Vert = 1}} x^\rmT A x.
\end{align*}
\end{lma}
\begin{proof}
    See Theorem 4.2.6 of \cite{horn2012matrix}.
\end{proof}

\begin{lma}
\label{lem: lambda_k(A) + lambda_1(B) < lambda_k(A+B) < lambda_k(A) + lambda_n(B)}
Let $A \in \BBR^{n \times n}$ and $B \in \BBR^{n \times n}$ be symmetric. Then, 
\begin{align*}
    \boldsymbol{\lambda_k}(A) + \eigmin(B) \le \boldsymbol{\lambda_k}(A+B) \le \boldsymbol{\lambda_k}(A) + \eigmax(B).
\end{align*}
\end{lma}
\begin{proof}
    See Theorem 10.4.11 of \cite{bernstein2009matrix}.
\end{proof}

\begin{lma}
\label{lem: eigenvalues of AB and BA}
Let $A \in \BBR^{n \times n}$ and $B \in \BBR^{n \times n}$. Then, the eigenvalues of $AB \in \BBR^{n \times n}$ are the eigenvalues of $BA \in \BBR^{n \times n}$
\end{lma}
\begin{proof}
   See Theorem 1.3.22 of \cite{horn2012matrix}.
\end{proof}

\begin{lma}
\label{lem: Chu 1998 Lemma 2.4, rank(H-M) = rank(H) - rank(M)}
Let $H \in \BBR^{n \times n}$ and $M\in \BBR^{n \times n}$ by symmetric matrices.
Then, 
$\rank(H-M) = \rank(H) - \rank(M)$ 
if and only if 
there exists $S\in \BBR^{n \times p}$ such that 
$M = HS(S^\rmT H S)^{-1} S^\rmT H$.
\end{lma}

\begin{proof}
    See Lemma 2.4 in \cite{Chu1998rank}.
\end{proof}

\begin{lma}
Let $H,M \in \BBR^{n \times n}$, where $H$ is positive definite and $M$ is positive semidefinite. 
%
%
%
Then, $H - M \succeq 0$ (respectively $H - M \succ 0$) if and only if $\rho (H^{-1}M) \le 1$ (respectively $\rho (H^{-1}M) < 1$). 
\label{lem: H-M PSD}
\end{lma}

\begin{proof}
    See Theorem 7.7.3 in \cite{horn2012matrix}.
\end{proof}

\begin{lma}

Let $A \in \BBR^{n \times n}$ be positive definite and $b \in \BBR^{n \times p}$.
Then, $b^\rmT A b \in \BBR^{p \times p}$ is positive semidefinite, $\rank(b^\rmT R b) = \rank(b)$, and $b^\rmT A b$ is positive definite if and only if $\rank(b) = p$.
\label{lem: phi R phi^T PSD}
\end{lma}

\begin{proof}
    See Observation 7.1.8 in \cite{horn2012matrix}.
\end{proof}

\section{Proofs in the Subspace Decomposition of a Positive-Definite Matrix}
\label{sec: subspace decomp proofs}


\begin{proof}[Proof of Theorem \ref{theo: Apar Aperp properties}]
By Lemma \ref{lem: phi R phi^T PSD}, $v^\rmT A v$ is positive definite. Hence, $A^{\parallel S}$ and $A^{\perp S}$ are well defined. 

\textit{Proof of (1a), (1b)}: Since $v^\rmT A v$ is positive definite, $(v^\rmT A v)^{-1}$ is also positive definite. Furthermore, Lemma \ref{lem: phi R phi^T PSD} implies that $A^{\parallel S} = (Av)(v^\rmT A v)^{-1} (Av)^\rmT $ is positive semidefinite. Next, note that by Lemma \ref{lem: H-M PSD}, $A^{\perp S} = A - A^{\parallel S}$ is positive semidefinite if and only if $\rho(A^{-1} A^{\parallel S}) \le 1$. 
Notice that 
\begin{align*}
    (A^{-1} A^{\parallel S})^2 &= \left[v (v^\rmT A v)^{-1} v^\rmT A\right]^2, \\
    &= v (v^\rmT A v)^{-1} v^\rmT A v (v^\rmT A v)^{-1} v^\rmT A, \\
    &= v (v^\rmT A v)^{-1} v^\rmT A, 
    = A^{-1} A^{\parallel S}.
\end{align*}
Therefore, $A^{-1} A^{\parallel S}$ is idempotent implying each of its $n$ eigenvalues are either $0$ or $1$. Hence, $\rho(A^{-1} A^{\parallel S}) \le 1$ implying that $A^{\perp S}$ is positive semidefinite.

\textit{Proof of (2a), (2b)}: Define $B \triangleq (v^\rmT A v)^{\nicefrac{1}{2}} \in \BBR^{p \times p}$. It follows that
\begin{align*}
    A^{\parallel S} = Av(B B^\rmT)^{-1}v^\rmT A = (A v B^{-1})^\rmT (A v B^{-1}).
\end{align*}
Thus, $\rank(A^{\parallel S}) = \rank(A v B^{-1})$. Since $A$ and $B^{-1}$ are positive definite, it follows that $\rank(A v B^{-1}) = \rank(v)$. Hence, $\rank(A^{\parallel S}) = \rank(v) = p$. Next, Lemma \ref{lem: Chu 1998 Lemma 2.4, rank(H-M) = rank(H) - rank(M)} implies that $\rank(A^{\perp S}) = \rank(A) - \rank(A^{\parallel S}) = n-p$. 

\textit{Proof of (3a), (3b)}: Let $x \in S$. Since $v_1,\hdots,v_p$ is a basis for $S$, it follows that there exists $c \in \BBR^{p \times 1}$ such that $x = vc$. Hence, 
\begin{align*}
    A^{\parallel S}x = Av(v^\rmT A v) v^\rmT A v c = A v c = Ax,
\end{align*}
and $A^{\perp S}x = (A-A^{\parallel S})x = Ax-Ax = 0$.
\end{proof}


\begin{proof}[Proof of Theorem \ref{theo: A^parallel uniqueness}]
%
Let $v_1, \cdots ,v_p \in \BBR^n$ be a basis for $S$ and $v \triangleq \begin{bmatrix} v_1 & \cdots & v_p \end{bmatrix} \in \BBR^{n \times p}$. Begin by writing $\tilde{A} = A X A$, where $X \triangleq A^{-1} \tilde{A} A^{-1} \in \BBR^{n \times n}$ and define $P \triangleq v(v^\rmT v)^{-1} v^\rmT \in \BBR^{n \times n}$. Since $\tilde{A}v = Av$, it follows that 
\begin{align*}
    XAP = A^{-1} \tilde{A} v(v^\rmT v)^{-1} v^\rmT = A^{-1} A v(v^\rmT v)^{-1} v^\rmT = P.
\end{align*}
Since $XAP = P$, it follows that the column space of $P$ is a subset of the column space of $X$. However, note that $\rank(P) = \rank(X) = p$. Therefore, $X$ and $P$ have the same column space.

Hence, for all $x \in \BBR^n$, there exists $y \in \BBR^n$ such that $Xx = Py$. Moreover, since $P^2 = P$, it follows that, for all $x \in \BBR^n$, there exists $y \in \BBR^n$ such that $PXx = PPy = Py = Xx$. Hence, $PX = X$. Moreover, since $P$, $X$, and $PX=X$ are symmetric, it follows that $P$ and $X$ commute. Therefore,
\begin{align*}
    PX = XP = X.
\end{align*}
Next, define $Q \triangleq PAP$. Note that $XQ = XPAP = XAP = P$. Moreover, since $X$, $Q$, and $XQ=P$ are symmetric, it follows that $X$ and $Q$ commute. Therefore,
\begin{align*}
    XQ = QX = P.
\end{align*}
Thus, it follows that $X$ and $Q$ satisfy the four Moore-Penrose Conditions:
\begin{enumerate}
    \item $XQX = XP = X$,
    \item $QXQ = QP = PAPP = PAP = Q$,
    \item $XQ = P = P^\rmT = (XQ)^\rmT$,
    \item $QX = P = P^\rmT = (QX)^\rmT$.
\end{enumerate}
Therefore, $X = Q^+$ and $\tilde{A} = A Q^+ A$. Finally, define $G \triangleq v(v^\rmT A v)^{-1}v^\rmT \in \BBR^{n \times n}$. It can be easily verified that $Q =v(v^\rmT v)^{-1} v^\rmT A v(v^\rmT v)^{-1} v^\rmT$ and $G$ satisfy the four Moore-Penrose Conditions $GQG = G$, $QGQ = Q$, $GQ = (GQ)^\rmT$, and $QG = (QG)^\rmT$.
%
%
Hence, $X = Q^+ = (G^+)^+ = G$, and $\tilde{A} = A G A = A v(v^\rmT A v)^{-1}v^\rmT A = A^{\parallel S}$.
\end{proof}
 

\begin{proof}[Proof of Proposition \ref{prop: duality}]
Note that since $\langle \cdot , \cdot \rangle_A$ is an inner product, it follows that $\dim(S^{\perp_A}) = n - \dim(S) = n - p$. Let $w_1,\hdots,w_{n-p} \in \BBR^n$ be a basis for $S^{\perp_A}$ and let $w \triangleq [w_1 \ \cdots \ w_{n-p} ] \in \BBR^{n \times n-p}$. Then, $A^{\perp (S^{\perp_A})}$ can be written as
\begin{align*}
    A^{\perp (S^{\perp_A})} = A - Aw(w^\rmT A w)^{-1} w^\rmT A.
\end{align*}
Note that $A^{\perp (AS)^\perp}$ is the difference of two symmetric matrices, $A$ and $Aw(w^\rmT A w)^{-1} w^\rmT A$, hence is symmetric. 
Next, since $\rank(w) = n-p$, it follows that $\rank(Aw(w^\rmT A w)^{-1} w^\rmT A) = n-p$. 
Furthermore, it follows from Lemma \ref{lem: Chu 1998 Lemma 2.4, rank(H-M) = rank(H) - rank(M)} that $\rank(A^{\perp (AS)^\perp}) = \rank(A) - \rank(Aw(w^\rmT A w)^{-1} w^\rmT A) = n - (n-p) = p$. 
Finally, since the columns of $w$ are elements of $S^{\perp_A}$, it follows that, for all $x \in S$, $w^\rmT A x= 0_{n-p \times 1}$. Hence, for all $x \in S$, 
\begin{align*}
    A^{\perp (S^{\perp_A})}x = Ax- Aw(w^\rmT A w)^{-1} w^\rmT A x = Ax.
\end{align*}
Thus, the three properties of Theorem \ref{theo: A^parallel uniqueness} are satisfied and $A^{\parallel S} = A^{\perp (S^{\perp_A})}$ holds, proving \eqref{eqn: AparS duality}.

Next, it follows from definition \eqref{eqn: Aperp definition} that $A^{\perp (S^{\perp_A})} = A-A^{\parallel (S^{\perp_A})}$ which can be rewritten as $A- A^{\perp (S^{\perp_A})} = A^{\parallel (S^{\perp_A})}$. 
It then follows from \eqref{eqn: AparS duality} that $A- A^{\parallel S} = A^{\parallel (S^{\perp_A})}$. 
Finally, by definition \eqref{eqn: Aperp definition}, it follows that $A^{\perp S} = A^{\parallel (S^{\perp_A})}$, proving \eqref{eqn: AperpS duality}.
\end{proof}

\section{Proof of Theorem \ref{theo: Rk upper and lower bounds}}
\label{sec: proof of covariance bounds}

We begin the proof of Theorem \ref{theo: Rk upper and lower bounds} with a useful Lemma. 

\begin{lem}
\label{lem: Rkperp eigenvalue bounds}
    For all $k \ge 0$, the $n-q_k$ nonzero eigenvalues of $R_k^\perp$ are bounded below by $\eigmin(R_k)$ and bounded above by $\eigmax(R_k)$.
\end{lem}

\textbf{Proof.}
    Let $k \ge 0$. To begin, note that by Theorem \ref{theo: Apar Aperp properties}, $\rank(R_k^\perp) = n - q_k$. 
    Therefore, $R_k^\perp$ has $n - q_k$ nonzero eigenvalues.
    
    First, we show the upper bound. Since $R_k = R_k^\parallel + R_k^\perp$, it follows from Lemma \ref{lem: lambda_k(A) + lambda_1(B) < lambda_k(A+B) < lambda_k(A) + lambda_n(B)} that $\eigmin(R_k^\parallel) + \eigmax(R_k^\perp) \le \eigmax(R_k)$. 
    Next, Theorem \ref{theo: Apar Aperp properties} implies that $R_k^\parallel$ is positive semidefinite, thus $0 \le \eigmin(R_k^\parallel)$. It then follows that $\eigmax(R_k^\perp) \le \eigmax(R_k)$.

    Second, to show the lower bound, first let $R_k^{\nicefrac{1}{2}} \in \BBR^n$ be the unique positive-definite matrix such that $R_k = R_k^{\nicefrac{1}{2}} R_k^{\nicefrac{1}{2}}$. 
    Furthermore, define $\Omega_k \in \BBR^{n \times n}$ as
    \begin{align*}
        \Omega_k \triangleq I_n - R_k^{\nicefrac{1}{2}}\bar{\phi}_k^\rmT (\bar{\phi}_k R_k \bar{\phi}_k^\rmT)^{-1} \bar{\phi}_k R_k^{\nicefrac{1}{2}}.
    \end{align*}
    It then follows that $R_k^\perp = R_k^{\nicefrac{1}{2}} \Omega_k R_k^{\frac{1}{2}}$. 
    It can also be easily verified that $\Omega_k^2 = \Omega_k$, hence $R_k^\perp = (R_k^{\frac{1}{2}} \Omega_k) (\Omega_k R_k^{\frac{1}{2}})$.
    Then, it follows from Lemma \ref{lem: eigenvalues of AB and BA} that the eigenvalues of $R_k^\perp$ are the same as the eigenvalues of $(\Omega_k R_k^{\frac{1}{2}}) (R_k^{\frac{1}{2}} \Omega_k) = \Omega_k R_k \Omega_k$.

    Next, since $R_k^\perp$ has $n-q_k$ nonzero eigenvalues, it suffices to show that $\boldsymbol{\lambda_{n - q_k}}(R_k^\perp) \ge \eigmin(R_k)$. Furthermore, by the previous result, it holds that $\boldsymbol{\lambda_{n - q_k}}(R_k^\perp) = \boldsymbol{\lambda_{n - q_k}}(\Omega_k R_k \Omega_k)$.
    It then follows from Lemma \ref{lem: min-max theorem} that
    \begin{align*}
        \boldsymbol{\lambda_{n - q_k}}(R_k^\perp) = \max_{\substack{S \subset \BBR^n \\ \dim(S) = n-q_k}} \min_{\substack{x \in S \\ \Vert x \Vert = 1}} x^\rmT \Omega_k R_k \Omega_k x.
    \end{align*}
    Furthermore, since $\Omega_k = R_k^{-\nicefrac{1}{2}} R_k^\perp R_k^{-\nicefrac{1}{2}}$ and $R^{-\nicefrac{1}{2}}$ is positive definite, it follows that $\rank(\Omega_k) = n - q_k$. Hence, $\dim(\mathcal{R}(\Omega_k)) = n - q_k$. Therefore,
    \begin{align*}
        \boldsymbol{\lambda_{n - q_k}}(R_k^\perp) \ge \min_{ x \in \mathcal{R}(\Omega_k), \ \Vert x \Vert = 1} x^\rmT \Omega_k R_k \Omega_k x.
    \end{align*}
    Next, since $\Omega_k^2 = \Omega_k$, it follows that, for all $x \in \mathcal{R}(\Omega_k)$, $\Omega_k x = x$. Hence,
    \begin{align*}
        \boldsymbol{\lambda_{n - q_k}}(R_k^\perp) \ge \min_{x \in \mathcal{R}(\Omega_k), \ \Vert x \Vert = 1} x^\rmT R_k x.
    \end{align*}
    Finally, note that
    \begin{align*}
        \min_{x \in \mathcal{R}(\Omega_k), \ \Vert x \Vert = 1} x^\rmT R_k x \ge \min_{\Vert x \Vert = 1} x^\rmT R_k x = \eigmin(R_k). \tag*{\mbox{$\blacksquare$}}
    \end{align*}

We now present the proof of Theorem \ref{theo: Rk upper and lower bounds}.

\begin{proof}[Proof of Theorem \ref{theo: Rk upper and lower bounds}]
    Proof follows by induction on $k \ge 0$. First, consider the base case $k = 0$. Note that $R_0 \succeq \eigmin(R_0) I_n$ implying \eqref{eqn: Rk lower bound} and that $R_0 \preceq \eigmax(R_0) I_n$ implying \eqref{eqn: Rk upper bound}.

    Next, let $k \ge 0$. 
    Let $\phi_k = U_k \Sigma_k V_k^\rmT$ be the full SVD of $\phi_k$, where $U_k \in \BBR^{p \times p}$ is orthogonal, $\Sigma_k \in \BBR^{p \times n}$ is rectangular diagonal with diagonal in descending order, and $V_k \in \BBR^{n \times n}$ is orthogonal.
    Moreover, denote the singular of $\phi_k$ as $\sigma_{1,k},\sigma_{2,k},\hdots,\sigma_{\min\{p,n\},k}$.
    Finally, since $q_k \le \min \{p,n\}$, we can partition the columns of $U_k$ and rows of $V_k$ as
    \begin{align}
        U_k = \begin{bmatrix}
            U_{1,k} & U_{2,k}
        \end{bmatrix},
        \quad
        V_k = \begin{bmatrix}
            V_{1,k} & V_{2,k}
        \end{bmatrix},
    \end{align}
    where $U_{1,k} \in \BBR^{p \times q_k}$ and $V_{1,k} \in \BBR^{n \times q_k}$ are the first $q_k$ columns of $U_k$ and $V_k$ respectively, and $U_{2,k} \in \BBR^{p \times (p-q_k)}$ and $V_{2,k} \in \BBR^{n \times (n-q_k)}$ are the remaining columns of $U_k$ and $V_k$ respectively.
    Note that $\bar{\phi}_k \in \BBR^{q_k \times n}$, defined in \eqref{eqn: phibar defn}, can be expressed as
    \begin{align}
        \bar{\phi}_k = U_{1,k}^\rmT \phi_k.
    \end{align}
    Moreover, the columns of $V_{1,k}$ form an orthonormal basis for $\mathcal{R}(\bar{\phi}_k^\rmT)$.

    Next, it follows from \eqref{eqn: Rkbar defn} and \eqref{eqn: SIFt Rk Update} that $R_{k+1} = \lambda R_k^\parallel + R_k^\perp + \bar{\phi}_k^\rmT \bar{\phi}_k$.
    It then follows from \eqref{eqn: Rk decomp} that
    \begin{align}
        R_{k+1} = \lambda R_k + (1-\lambda) R_k^\perp + \bar{\phi}_k^\rmT \bar{\phi}_k.
    \end{align}
    Pre-multiplying by $V_k^\rmT$ and post-multiplying by $V_k$, it follows that
    \begin{align}
    \label{eqn: V R V update temp}
        V_k^\rmT R_{k+1} V_k = V_k^\rmT [\lambda R_k + (1-\lambda) R_k^\perp + \bar{\phi}_k^\rmT \bar{\phi}_k] V_k.
    \end{align}
    Since the columns of $V_{1,k}$ form an orthonormal basis for $\mathcal{R}(\bar{\phi}_k^\rmT)$, it follows from Theorem \ref{theo: Apar Aperp properties} that $R_k^\perp V_{1,k} = 0_{n \times q_k}$.
    It then follows that
    \begin{align}
    \label{eqn: QRperpQ^T temp}
        V_k^\rmT R_k^\perp V_k =
        \begin{bmatrix}
       0_{q_k \times q_k} & 0_{q_k \times (n-q_k)}
        \\ 
        0_{(n-q_k) \times q_k} & V_{2,k}^\rmT R_k^\perp V_{2,k}
        \end{bmatrix}.
    \end{align}
    Moreover, note that the columns of $V_{2,k}$ are orthogonal to the columns of $V_{1,k}$, hence are also orthogonal to the rows of $\bar{\phi}_k$. 
    Therefore, $\bar{\phi}_k V_{2,k} = 0_{q_k \times (n-q_k)}$.
    Thus it follows that
    \begin{align}
    \label{eqn: phik V2k = 0 temp}
        V_k^\rmT \bar{\phi}_k^\rmT \bar{\phi}_k V_k = 
        \begin{bmatrix}
        V_{1,k}^\rmT \bar{\phi}_k^\rmT \bar{\phi}_k V_{1,k} & 0_{q_k \times (n-q_k)} \\ 0_{(n-q_k) \times q_k} & 0_{(n-q_k) \times (n-q_k)}
        \end{bmatrix}.
    \end{align}
    Substituting \eqref{eqn: QRperpQ^T temp} and \eqref{eqn: phik V2k = 0 temp} into \eqref{eqn: V R V update temp}, it follows that 
    \begin{align}
    \label{eqn: Q Rk+1 Q^T temp 2}
        V_k^\rmT R_{k+1} & V_k  =  \lambda V_k^\rmT R_{k} V_k \nonumber
        \\
        & + \begin{bmatrix}
        V_{1,k}^\rmT \bar{\phi}_k^\rmT \bar{\phi}_k V_{1,k} & 0_{q_k \times (n-q_k)} \\ 0_{(n-q_k) \times q_k} & (1-\lambda) V_{2,k}^\rmT R_k^\perp V_{2,k}
        \end{bmatrix}.
    \end{align}
    We now prove two claims which will later be useful.
    
    \textbf{\textit{Claim 1:}} $V_{1,k}^\rmT \bar{\phi}_k^\rmT \bar{\phi}_k V_{1,k} = \diag(\sigma_{1,k}^2,\hdots,\sigma_{q_k,k}^2)$.

    \textbf{\textit{Proof of Claim 1:}}
    Note that $\bar{\phi}_k = U_{1,k}^\rmT U_k \Sigma_k V_k^\rmT$.
    Furthermore, note that, $U_{1,k}^\rmT U_k = [U_{1,k}^\rmT U_{1,k} \ U_{1,k}^\rmT U_{2,k} ] = [I_{q_k} \ 0_{q_k \times (p-q_k)}]$.
    Similarly, we can express $V_{1,k}^\rmT V_k =  [V_{1,k}^\rmT V_{1,k} \ V_{1,k}^\rmT V_{2,k} ] = [I_{q_k} \ 0_{q_k \times (n-q_k)}]$. 
    Therefore, $V_{1,k}^\rmT \bar{\phi}_k^\rmT \bar{\phi}_k V_{1,k}$ can be written as
    \begin{align}
        V_{1,k}^\rmT \bar{\phi}_k^\rmT \bar{\phi}_k V_{1,k} 
        &= V_{1,k}^\rmT V_k \Sigma_k^\rmT U_k^\rmT U_{1,k} U_{1,k}^\rmT U_k \Sigma_k V_k^\rmT V_{1,k} \nonumber
        \\
        &= \begin{bmatrix}
            I_{q_k} & 0 
        \end{bmatrix} 
        \Sigma_k^\rmT
        \begin{bmatrix}
            I_{q_k} & 0 \\ 0 & 0
        \end{bmatrix}
        \Sigma_k
        \begin{bmatrix}
            I_{q_k} \\ 0 
        \end{bmatrix}
        \nonumber
        \\
        & = \diag(\sigma_{1,k}^2,\hdots,\sigma_{q_k,k}^2),
    \end{align}
    with zero matrices of appropriate sizes. \hfill$\square$
    
    \textbf{\textit{Claim 2:}} \textit{The $n-q_k$ nonzero eigenvalues of $R_k^\perp$ are the eigenvalues of $V_{2,k}^\rmT R_k^\perp V_{2,k} \in \BBR^{(n-q_k)\times(n-q_k)}$.}
    
    \textbf{\textit{Proof of Claim 2:}}
    First, note that, by Theorem \ref{theo: Apar Aperp properties}, $\rank(R_k^\perp) = n-q_k$, hence $R_k^\perp$ has $n-q_k$ nonzero eigenvalues.
    Next, it follows from \eqref{eqn: QRperpQ^T temp} that $V_k^\rmT R_k^\perp V_k$ is block diagonal. 
    Therefore, the eigenvalues of $V_k^\rmT R_k^\perp V_k$ are $q_k$ zeros and the $n-q_k$ eigenvalues of $V_{2,k}^\rmT R_k^\perp V_{2,k}$.
    Finally, since similarity transformation preserve eigenvalues, it follows that the eigenvalues of $R_k^\perp$ are $q_k$ zeros and the $n-q_k$ eigenvalues of $V_{2,k}^\rmT R_k^\perp V_{2,k}$. \hfill$\square$

    We now show the induction step for statement \textit{(1)} of Theorem \ref{theo: Rk upper and lower bounds}.
    First, assume, for inductive hypothesis, that \eqref{eqn: Rk lower bound} holds $k \ge 0$. 
    Applying Lemma \ref{lem: lambda_k(A) + lambda_1(B) < lambda_k(A+B) < lambda_k(A) + lambda_n(B)} to \eqref{eqn: Q Rk+1 Q^T temp 2} yields
    \begin{align}
    \label{eqn: eigmin Rk+1 temp 1}
        & \eigmin ( V_k^\rmT R_{k+1} V_k) \ge  \lambda \eigmin(V_k^\rmT R_k V_k) + \nonumber \\
        & \min \{ \eigmin(V_{1,k}^\rmT \bar{\phi}_k^\rmT \bar{\phi}_k V_{1,k}) , (1-\lambda) \eigmin(V_{2,k}^\rmT R_k^\perp V_{2,k}) \}.
    \end{align}
    Since similarity transformations preserve eigenvalues, it then follows that 
    \begin{align}
    \label{eqn: eigmin Rk+1 temp}
        & \eigmin ( R_{k+1}) \ge  \lambda \eigmin(R_k) + \nonumber \\
        & \min \{ \eigmin(V_{1,k}^\rmT \bar{\phi}_k^\rmT \bar{\phi}_k V_{1,k}) , (1-\lambda) \eigmin(V_{2,k}^\rmT R_k^\perp V_{2,k}) \}.
    \end{align}
    Next, applying the results of Claims 1 and 2, it follows that
    \begin{align}
    \label{eqn: eigmin Rk+1 temp 2}
        \eigmin ( R_{k+1}) & \ge  \lambda \eigmin(R_k) \nonumber \\
        & + \min \{ \sigma_{q_k,k}^2 , (1-\lambda) \boldsymbol{\lambda_{n-q_k}}(R_k^\perp) \}.
    \end{align}
    Note that, by construction of $q_k$, it follows that $\sigma_{q_k,k} \ge \sqrt{\varepsilon}$.
    Moreover, by Lemma \ref{lem: Rkperp eigenvalue bounds}, $\boldsymbol{\lambda_{n-q_k}}(R_k^\perp) \ge \eigmin(R_k)$.
    Applying these two inequalities to \eqref{eqn: eigmin Rk+1 temp 2}, it follows that
    \begin{align}
    \label{eqn: eigmin Rk+1 temp 2.5}
        \eigmin ( R&_{k+1}) \ge  \lambda \eigmin(R_k) 
        + \min \{ \varepsilon , (1-\lambda) \eigmin(R_k) \} \nonumber
        \\
        &= \min \{ \varepsilon + \lambda \eigmin(R_k) , \eigmin(R_k) \}.
    \end{align}
    Next, substituting inductive hypothesis \eqref{eqn: Rk lower bound} into \eqref{eqn: eigmin Rk+1 temp 2.5}, it follows that
    \begin{align}
    \label{eqn: eigmin Rk+1 temp 3}
        \eigmin ( R_{k+1}) \ge \min \Big\{ 
        &
        \varepsilon + \lambda (\frac{\varepsilon}{1-\lambda}), 
        \varepsilon + \lambda \eigmin(R_0), \nonumber
        \\
        & \frac{\varepsilon}{1-\lambda},
        \eigmin(R_0)
        \Big\}.
    \end{align}
    Note that $\varepsilon + \lambda (\frac{\varepsilon}{1-\lambda})$ simplifies to $\frac{\varepsilon}{1-\lambda}$.
    Hence, \eqref{eqn: eigmin Rk+1 temp 3} can be rewritten as
    \begin{align}
    \label{eqn: eigmin Rk+1 temp 4}
        \eigmin ( R_{k+1}) \ge \min \{ 
        \varepsilon + \lambda \eigmin(R_0), 
        \frac{\varepsilon}{1-\lambda},
        \eigmin(R_0)
        \}.
    \end{align}
    In the case where $\eigmin(R_0) \ge \frac{\varepsilon}{1-\lambda}$, it follows that 
    \begin{align*}
        \varepsilon + \lambda \eigmin(R_0) \ge \varepsilon + \lambda (\frac{\varepsilon}{1-\lambda}) = (\frac{\varepsilon}{1-\lambda}).
    \end{align*}
    In the case where $\eigmin(R_0) < \frac{\varepsilon}{1-\lambda}$, it follows that $\varepsilon - (1-\lambda) \eigmin(R_0) > 0$, and hence
    \begin{align*}
        \varepsilon + \lambda \eigmin(R_0) &= \eigmin(R_0) + \varepsilon - (1-\lambda) \eigmin(R_0) 
        \\
        &> \eigmin(R_0).
    \end{align*}
    Combining these two cases implies that
    \begin{align}
    \label{eqn: eps + lambda eigmin(R0) temp}
        \varepsilon + \lambda \eigmin(R_0) \ge \min \{ \frac{\varepsilon}{1-\lambda}, \eigmin(R_0) \}.
    \end{align}
    Substituting \eqref{eqn: eps + lambda eigmin(R0) temp} into \eqref{eqn: eigmin Rk+1 temp 4} yields
    \begin{align}
        \eigmin ( R_{k+1}) \ge \min \{ \frac{\varepsilon}{1-\lambda}, \eigmin(R_0) \}.
    \end{align}
    Therefore, by induction, \eqref{eqn: Rk lower bound} holds for all $k \ge 0$.

    Next, we will show the induction step for statement \textit{(2)} of Theorem \ref{theo: Rk upper and lower bounds}.
    Assume that $(\phi_k)_{k=0}^\infty$ is bounded with upper bound $\beta \in (0,\infty)$ and assume, for inductive hypothesis, that \eqref{eqn: Rk upper bound} holds for $k \ge 0$. 
    Next, applying Lemma \ref{lem: lambda_k(A) + lambda_1(B) < lambda_k(A+B) < lambda_k(A) + lambda_n(B)} to \eqref{eqn: Q Rk+1 Q^T temp 2}, it follows that
    \begin{align}
    \label{eqn: eigmax Rk+1 temp}
        & \eigmax (V_k^\rmT R_{k+1} V_k) \le  \lambda \eigmax(V_k^\rmT R_k V_k) + \nonumber \\
        & \max \{ \eigmax(V_{1,k}^\rmT \bar{\phi}_k^\rmT \bar{\phi}_k V_{1,k}) , (1-\lambda) \eigmax(V_{2,k}^\rmT R_k^\perp V_{2,k}) \}.
    \end{align}
    By analogous reasoning to that used previously, we use the fact that similarity transformations preserve eigenvalues, Lemma \ref{lem: Rkperp eigenvalue bounds}, and Claims 1 and 2 to simplify \eqref{eqn: eigmax Rk+1 temp} to
    \begin{align}
    \label{eqn: eigmax Rk+1 temp 1.5}
        \eigmax ( R_{k+1}) &\le  \lambda \eigmax(R_k) \nonumber
        \\
        &+ \max \{ \sigma_{1,k}^2 , (1-\lambda) \eigmax(R_k) \}.
    \end{align}
    Moreover since $(\phi_k)_{k=0}^\infty$ is bounded with upper bound $\beta$, it follows from Definition \ref{defn: sequence upper bounded} that $\sigma_{1,k}^2 \le \beta$. Thus, it follows from \eqref{eqn: eigmax Rk+1 temp 1.5} that
    \begin{align}
    \label{eqn: eigmax Rk+1 temp 2}
        \eigmax ( R_{k+1})
        &\le \max \{ \beta + \lambda \eigmax(R_k) , \eigmax(R_k) \}.
    \end{align}
    Next, substituting inductive hypothesis \eqref{eqn: Rk upper bound} into \eqref{eqn: eigmax Rk+1 temp 2} yields
    \begin{align}
    \label{eqn: eigmax Rk+1 temp 3}
        \eigmax ( R_{k+1}) \le \max \Big\{ 
        &
        \beta + \lambda (\frac{\beta}{1-\lambda}), 
        \beta + \lambda \eigmax(R_0), \nonumber
        \\
        & \frac{\beta}{1-\lambda},
        \eigmax(R_0)
        \Big\}.
    \end{align}
    Then, by similar reasoning to before, it can be shown that \eqref{eqn: eigmax Rk+1 temp 3} simplifies to
    \begin{align}
        \eigmax ( R_{k+1}) \le \max \{ \frac{\beta}{1-\lambda}, \eigmax(R_0) \}.
    \end{align}
    Therefore, by induction, \eqref{eqn: Rk upper bound} holds for all $k \ge 0$. 
\end{proof}

\section{Comparison to \cite{zhu2021recursive}}
\label{sec: issues with oblique}
In \cite{zhu2021recursive}, a different vector measurement extension of \cite{cao2000directional} is presented using an oblique projection decomposition. A tuning parameter $\varepsilon > 0$ is chosen and, for all $k \ge 0$, if $\Vert \phi_k \Vert_2 \ge \varepsilon$, then the information matrix $R_k \in \BBR^{n \times n}$ is updated as
\begin{equation}
\label{eqn: Rk update zhu 2021}
    R_{k+1} = R_k - (1-\lambda) R_k \phi_k^\rmT (\phi_k R_k \phi_k^\rmT)^+ \phi_k R_k + \phi_k^\rmT \phi_k,
\end{equation}
and if $\Vert \phi_k \Vert_2 < \varepsilon$, then $R_{k+1} = R_k + \phi_k^\rmT \phi_k$.
It is shown in Theorem 3 of \cite{zhu2021recursive} that, for all $k \ge 0$, $R_k$ is positive definite.
However, we show that there may not exist $\gamma > 0$ such that, for all $k \ge 0$, $\eigmin(R_k) \succeq \gamma$.
In particular, the following example shows that $R_k$ may converge to a positive semidefinite matrix.
Hence, Theorem \ref{theo: Rk upper and lower bounds} of SIFt-RLS provides stronger bounds on the eigenvalues of the information matrix than \cite{zhu2021recursive} does.

Let $\varepsilon > 0$. Let $R_0 = \diag(\varepsilon^2,1)$ and, for all $k \ge 0$, let $\phi_k = \diag(\varepsilon,\lambda^{\frac{k+1}{2}})$. Then, for all $k \ge 0$, $\Vert \phi_k \Vert_2 \ge \varepsilon$. Moreover, for all $k \ge 0$, $\phi_k$ is nonsingular, hence $\eqref{eqn: Rk update zhu 2021}$ simplifies to 
\begin{align}
    R_{k+1} = \lambda R_k + \phi_k^\rmT \phi_k = \lambda R_k + \begin{bmatrix}
        \varepsilon^2 & 0 \\
        0 & \lambda^{k+1}
    \end{bmatrix}.
\end{align}
Therefore, for all $k \ge 0$, 
\begin{align}
    R_k = \begin{bmatrix}
        \varepsilon^2 \sum_{i=0}^k \lambda^i & 0 \\
        0 & (k+1) \lambda^k
    \end{bmatrix}.
\end{align}
This implies that $\lim_{k \rightarrow \infty} R_k = \diag(\frac{\varepsilon^2}{1-\lambda},0)$.
Therefore, there does not exist $\gamma > 0$ such that, for all $k \ge 0$, $\eigmin(R_k) \succeq \beta$.

\end{document}